\documentclass[11pt,a4paper,amssymb,amsmath, tightenlines]{article}

\usepackage{amsmath, amssymb, amsthm, latexsym, mathrsfs, url, color, epsfig, graphics, float, setspace,hyperref, graphicx, bbm, relsize}
\usepackage{amsfonts}
\usepackage{tikz}
\usepackage{sectsty}
\usepackage{float}
\usepackage[T1]{fontenc}
\usepackage[bitstream-charter]{mathdesign}
\usepackage{pgfplots}
\usepackage{color}
\usepackage{bm}

\usepackage[numbers]{natbib}
\usepackage{algorithmicx}
\usepackage{algorithm}
\usepackage{algpseudocode}
\usepackage{appendix}

\usepackage{enumerate}
\usepackage{tikz}
\usepackage[utf8]{inputenc}
\usepackage{pgfplots} 
\usepackage{pgfgantt}
\usepackage{pdflscape}
\newlength\fwidth
\newlength\fheight

\usepackage{caption}
\usepackage{subcaption}

\usepackage{pdflscape}
\usepackage{afterpage}
\usepackage{rotating}
\usepackage{array, makecell}
\usepackage{multirow}



\pgfplotsset{compat=newest} 
\pgfplotsset{plot coordinates/math parser=false}

\allowdisplaybreaks[1]

\newtheorem{Definition}{Definition}[section]
\newtheorem{rem}{Remark}
\newtheorem{Proposition}{Proposition}[section]

\newtheorem{Corollary}{Corollary}[section]
\newtheorem{Example}{Example}[section]
\numberwithin{equation}{section}

\newcommand{\e}{{\rm e}}
\newcommand{\lnn}{{\rm ln}}
\newcommand{\expp}{{\rm exp}}
\newcommand{\erf}{{\rm erf}}

\newcommand{\F}{\mathscr{F}}

\newcommand{\PR}{\mathbb{P}}

\newcommand{\rd}{\textup{d}}

\newcommand{\indi}[1]{1\hspace{-.09cm}\textup{\textrm{l}}}

\begin{document}
\title{\vspace{-2cm}\bf Stochastic measure distortions induced by quantile processes for risk quantification and valuation}
\author{Holly Brannelly$^{\dag}$, Andrea Macrina{$^{\dag\, \ddag}$\footnote{Corresponding author: a.macrina@ucl.ac.uk}\ }, Gareth W. Peters{$^{\S}$} \\ \\ {$^{\dag}$Department of Mathematics, University College London} \\ {London, United Kingdom} 
\\ {$^{\ddag}$African Institute of Financial Markets \& Risk Management} \\ {University of Cape Town} \\ {Rondebosch 7701, South Africa} 
\\ {$^{\S}$}Department of Statistics \& Applied Probability \\ University of California Santa Barbara \\Santa Barbara, USA}
\date{\today}
\maketitle
\vspace{-1cm}
\begin{abstract}
\noindent 
We develop a novel stochastic valuation and premium calculation principle based on probability measure distortions that are induced by quantile processes in continuous time. Necessary and sufficient conditions are derived under which the quantile processes satisfy first-- and second--order stochastic dominance. The introduced valuation principle relies on stochastic ordering so that the valuation risk--loading, and thus risk premiums, generated by the measure distortion is an ordered parametric family. The quantile processes are generated by a composite map consisting of a distribution and a quantile function. The distribution function accounts for model risk in relation to the empirical distribution of the risk process, while the quantile function models the response to the risk source as perceived by, e.g., a market agent. This gives rise to a system of subjective probability measures that indexes a stochastic valuation principle susceptible to probability measure distortions. We use the Tukey--$gh$ family of quantile processes driven by Brownian motion in an example that demonstrates stochastic ordering. We consider the conditional expectation under the distorted measure as a member of the time--consistent class of dynamic valuation principles, and extend it to the setting where the driving risk process is multivariate. This requires the introduction of a copula function in the composite map for the construction of quantile processes, which presents another new element in the risk quantification and modelling framework based on probability measure distortions induced by quantile processes.   
\\\vspace{-1cm}\\
\\
\\
{\bf Keywords}: Quantile processes; stochastic ordering; distortion of probability measures; distortion-based pricing; stochastic valuation principle and premium calculation; dynamic risk-loading; Tukey gh-transform; Radon-Nikodym derivative; skewness, kurtosis, and further higher moments. 
\\\vspace{-0.25cm}
\end{abstract}
\section{Introduction}
In this work we propose a general, time--consistent, and dynamic approach for the pricing of risks, or valuation of loss premiums, in markets where arbitrage pricing may fall short due to the (possible) incompleteness of markets, and so payoffs of the considered risks cannot be replicated in terms of traded securities. In order to achieve this, we develop extensions of the broad class of continuous--time quantile diffusion first presented in \cite{BMPqd1}. These quantile diffusions are developed via a composite map that distorts the finite--dimensional distributions of a driving diffusion. The construction leads to a large family of flexible output processes on the quantile space (function space of quantile functions, see \cite{BMPqd1}), accommodating a wide range of tail behaviours in the finite--dimensional distributions, and the higher--order moments are both directly parameterised and interpretable. Here, the distortion map employed consists of the composition of a distribution function and a quantile function, analogous---but distinct---to a transmutation map, see \cite{gilchrist,peters2016estimating,quantiletutorial,distalchemy, rtm, quantilemech2, quantilemech}. This construction allows one to alter the base distribution in such a way to better satisfy some target model or objective.  The ideas presented broaden the existing quantile modelling literature, which is dense in the areas of statistical regression and time series econometrics, stochastic processes, risk management and insurance, and mathematical statistics within the study of empirical processes.  The reader may refer to \cite{conditionalqp,dqm,qp1,qp2, qp3,dassios, embrechts1995proof,gilchrist, koenkerbassett, koenkerqr,quantilear, quantilemech} and the references therein, as well as \cite{quantiletutorial} for a tutorial review on the use of quantile functions in discrete, time series based statistical risk modelling. In this paper we extend quantile diffusions to consider general continuous--time processes with c\`adl\`ag paths, thus producing a wider class of quantile processes which, under certain conditions on the distortion map, will have c\`adl\`ag paths.  The class of quantile processes with c\`adl\`ag paths is then studied from the perspective of the distorted probability measures they induce, in relation to the choice of composite map and its interpretation with regard to risk and the risk preferences of market agents.  

Achieving dynamic distortions in the quantile space is advantageous in the context of risk management, among others, where problems involving the quantification of risk are often expressed in terms of quantile functions or quantile processes rather than a density.  We propose a general, time--consistent, and dynamic approach for the pricing of risks (or valuation of loss premiums) in (possibly) incomplete markets. In such a context, arbitrage pricing theory may fall short, because payoffs contingent on considered risk sources cannot be replicated with traded financial instruments. We base our approach on probability measures induced by quantile processes, whereby such ``subjective" probability measures inherit the (intentionally constructed and directly parameterised) statistical properties of the quantile process. The manner in which we construct probability distortions is a flexible and unrestricted alternative to commonly used distortion--based pricing techniques such as those produced using distortion operators, see, e.g., \cite{godin2012, godin2019,kijimamuromachi, wang, wang2,wang1996,wang1996ordering}, or a weighting/exponential tilting, see, e.g., \cite{esscher1932,furman2008weighted,cruz2015fundamental} and references therein.  Unlike these existing methods, the approach developed in this work facilitates the direct parameterisation of features, such as skewness and kurtosis of the risk drivers, which thereby allows one to incorporate knowledge of asymmetries or tail risk, and to directly quantify their impact on the resultant valuations. As such, of key importance in this work is the explication of how the choice of the composite map affects the behaviour of prices and premiums obtained under the stochastic valuation principle. In this regard, in Section \ref{stochasticorderingsection}, we derive the necessary and sufficient conditions under which the constructed quantile processes satisfy first-- and second--order stochastic dominance, leading to ordered parametric families of prices or premiums related to some given base risk.  The stochastic ordering results can also be utilised to ensure the valuation principle induces consistent risk--loading and thus is axiomatically sound in the context of insurance prices and premium calculation principles.  

In Section \ref{valuationprinciplesection}, a general and time--consistent stochastic valuation principle is defined and the connection with dynamic, convex risk measures, see e.g., \cite{acciaio2011dynamic, bion2006dynamic,bion2008dynamic, bion2009time,detlefsen2005conditional}, is made.   As such, the general framework is shown to incorporate well--known pricing frameworks built by concave distortion operators, see above references, and a large number of premium calculation principles (PCPs), see \cite{laeven2008premium} for a detailed survey on PCPs.  A special case of the stochastic valuation principle is then specified, employing a conditional expectation under the distorted probability measure induced by the quantile process, thus automatically ensuring time--consistency. The choice of the probability measure, determined by the choice of the quantile distortion map, determines the behaviour of the output price process for any given input risk. For example, if the input risk process becomes more heavy--tailed or skewed, or both, under the distortion map, then the valuation principle induces a risk--loading.  We construct the composite map so to incorporate financial market risk, characteristics directly associated with the underlying risk process that may not be accounted for when considering historical losses (e.g., technological advancements, changes in climate, geopolitical turmoil), and investor risk preferences along with their perception of the risk source.  The latter of these factors is emphasised throughout, and as such, when the Radon--Nikodym derivative between the real--world and the distorted probability measure is derived and the connection with the pricing kernel is made, this is equivalent to the risk preferences of each market participant corresponding to a different, and not necessarily unique, pricing kernel.  This is of course consistent with the assumption of market incompleteness. In the insurance setting, the framework presents a fairly general class of axiomatically justified premium principles with high levels of flexibility. These premium principles allow one to capture skewness and leptokurtosis that is commonly observed in markets, and they can be used to explicitly incorporate more structure into subjective assignments of elicitable information regarding an investor's risk preferences.  This is captured through the induced characteristics of the quantile process, regardless of whether they are tractable or not.  

In the decision theory literature, it is well--established to consider a market agent’s preferences through a utility function, where actuarial prices of insurance--based risks are calculated by taking expectations under some probability measure.  In expected utility theory \cite{vnm}, the utility function is subjective and the probability measure is objective, whereas subjective expected utility theory \cite{savage1954} allows one to also consider a subjective probability measure, based on the considered decision situation or market conditions. Such subjective decision making approaches also arise in aspects of risk management when one defines scenario--based risk measure evaluation, see \cite{acerbi2002spectral,artzner1999coherent}.   A review on non--expected utility theories, including rank--dependent expected utility \cite{quiggin1982theory, yaari1987} and Choquet expected utility theory (CEU) \cite{choquet1954theory,schmeidler1989subjective}, is given in \cite{sugden1997alternatives}.  A connection between the functional form of the valuation principle presented in this work and utility functions, e.g., von Neumann-Morgenstern \cite{vnm} utility functions, as well as the value function underlying prospect theory \cite{kahneman1979prospect} is discussed. We also draw the reader's attention back to the stochastic ordering of risks and the important role it plays in decision making.  In decision theory---specifically, expected utility theory---the maximisation of von Neumann-Morgenstern \cite{vnm} non--decreasing (non--decreasing and concave) utility functions by rational market participants equates to the notion of first--order (second--order) stochastic dominance.  Therefore, valuation frameworks with preserved stochastic ordering, of some order, imply monotonic behaviour between the level of risk and the price associated to some contract contingent on the risk, e.g., an insurance premium.  That is, assuming investors act rationally, it is desirable for valuation frameworks to preserve the property of stochastically--ordered risks and so, in this paper, we emphasise such results in the context of the stochastic valuation principle induced by quantile processes. 

In Section \ref{multiqpsection}, we extend the valuation principle induced by quantile processes to the multivariate setting where (auxiliary) risk factors, e.g., other highly correlated risks or external macroeconomic/systemic risks, may also be considered.  Here, a multivariate driving process and a copula are employed to construct a univariate quantile process that induces the distorted measure that characterises the stochastic valuation principle.  As a result, the composite map, which now consists of a quantile function and a copula function, and each marginal of the multivariate risk factor impact the distorted measure induced by the quantile process, and thus how risks are valued when taking conditional expectations under this measure.  The idea is that the incorporation of such external risks in the construction of the distorted probability measure refines the accuracy of the valuation principle.  Here, one or more marginals of the multivariate driving process are the underlying risks on which the financial or insurance contract is written, and the remaining marginals of the driving process model the external risks that are accounted for in the valuation problem.  The model is presented in the context of pricing insurance layer and stop--loss contracts.  We envisage applications of this model to include, e.g., the pricing of insurance contracts written on some loss process where external risks related to climate change are likely to impact. While the external risk sources may trigger the payout of an insurance contract, such risks are not regarded as the directly insurable loss, nor is their impact prevalent in the historical loss data of the risk being insured.  As the realised impact of climate change is on the rise, it becomes more important for insurers to factor such risks into their pricing models.  Here, the climate risk processes are taken as marginals of the multivariate driver used in the construction of the quantile process that induces the distorted probability measure.  The composite map (consisting of a quantile function and a copula) can also be utilised to produce excess skewness or kurtosis that the insurer deems reasonable to factor into the model but is not captured in historical loss data, or the risk preferences of the agent buying the insurance contract.  It may be at the discretion of the insurer to determine the appropriate composite map used in the construction of the quantile process.  A theoretically similar idea is formalised in \cite{zhu2019agricultural}, where a premium principle for agricultural losses is constructed by incorporating a reweighing of the historical losses with systemic risk factors---here, the insurer may select the appropriate weighting function.  The discussion on the multidimensional extension to the quantile process--induced valuation framework concludes this paper.    
\section{Stochastic quantile processes in continuous time}\label{qpdefnsection}
A class of quantile diffusions, which are generated by a composite map consisting of a distribution function and a quantile function, have been developed in \cite{BMPqd1}. In this section, we extend this class by constructing quantile processes driven by more general processes, thus including L\'evy processes as stochastic drivers, for example.   
\

We consider the probability space $(\Omega, \F, \PR)$ equipped with a filtration $(\F_t)_{t\geq 0}$, where the time line is given by $t\in[0,\infty)$. We assume the existence of a money--market, or ``risk--free'', asset with price process $(B_t)_{t\in[0,\infty)}$ that offers a positive rate of return. For the construction of the quantile processes, we will employ an $(\F_t)$-adapted process $(Y_t)_{t\in[0,\infty)}$ with  c\`adl\`ag paths. First, however, we recall the definition of the generalized inverse function, see \cite{geninverse}, in terms of which the quantile function of a random variable $X$ is given. For an increasing function $F:\mathbb{R}\rightarrow\mathbb{R}$ with $F(-\infty)=\lim_{x\downarrow -\infty}F(x)$ and $F(\infty)=\lim_{x\uparrow\infty}F(x)$, the collection $\mathcal{Q}$ of generalised inverse functions $Q:=F^{-}:\mathbb{R}\rightarrow[-\infty,\infty]$ of $F$ is defined by
\begin{equation}\label{geninverse}
Q(y)=\inf\left\{x\in\mathbb{R}:F(x)\geq y \right\}, \hspace{2mm} y\in\mathbb{R},
\end{equation}
with the convention that $\inf\emptyset=\infty$. For $X$ a real--valued random variable with distribution function $F_X:\mathbb{R}\rightarrow[0,1]$, the quantile function of $X$ is $Q_X=F_X^{-}:[0,1]\rightarrow[-\infty,\infty]$ where $Q_X\in\mathcal{Q}$. We are now in the position to construct the following class of quantile processes.
\begin{Definition}\label{processdrivendefn}
Let $Q_{\zeta}(u;\bm{\xi})$ be a quantile function, where $u\in[0,1]$ is the quantile level and $\bm{\xi}\in\mathbb{R}^d$ is a $d$-dimensional vector of parameters. For $t\in(0,\infty)$, let $F(t,y;\bm{\theta}):\mathbb{R}^+\times\mathbb{R}\rightarrow[0,1]$ be a continuous distribution function, where $\bm{\theta}\in\mathbb{R}^{d^\prime}$ is a $d^\prime$--dimensional vector of parameters. Consider a process $(Y_t)_{t\in[0,\infty)}$ with continuous paths.  At each time $t\in(0,\infty)$, the random--level quantile process is defined by
\begin{equation}\label{processdrivenQP}
Z_t = Q_{\zeta}\left(F(t,Y_t;\bm{\theta});\bm{\xi}\right),
\end{equation}
that is, $Z_t=Q_{\zeta}(F(t,Y(t,\omega);\bm{\theta});\bm{\xi}):\mathbb{R}^+\times(\mathbb{R}^+\times\Omega)\rightarrow[-\infty,\infty]$, and the map $t\mapsto Z(t,\omega)$ for all $t\in(0,\infty)$ is $\mathscr{F}_t$--measurable. Here, the index $\zeta$ characterises the family of quantile processes to which $(Z_t)$ belongs.  
\end{Definition}

Here, if $F(t,y;\bm{\theta})=F_Y(t,y;\bm{\theta})$ is the finite dimensional distribution that governs each marginal of the process $(Y_t)$ for $t\in(0,\infty)$, then the marginals of the process given by $U_t:=F_Y(t,Y_t;\bm{\theta})$ are uniformly distributed on $[0,1]$.  We say that $F_Y(t,y;\bm{\theta})_{t\in(0,\infty)}$ is the `{\it true law}' of $(Y_t)_{t\in[0,\infty)}$.  Though in this paper we will focus on quantile processes with continuous paths, the composite map (\ref{processdrivenQP}) can be used to obtain quantile processes with discrete paths.  It suffices that either $(Y_t)$ is a c{\`a}dl{\`a}g process with discrete (discontinuous) paths, $F(t,y;\bm{\theta}):\mathbb{R}^+\times \mathscr{D}\rightarrow \mathscr{U}$ is a discrete distribution function with $\mathscr{D}$ and $\mathscr{U}$ the collections of all countable subsets of $\mathbb{R}$ and $[0,1]$, respectively, or that the function $Q_{\zeta}(u;\bm{\xi})$ might be the quantile function of a discrete random variable for $(Z_t)$ to have discrete paths.  The quantile process $(Z_t)$ will be given by the composite map (\ref{processdrivenQP}), where $Z_t:\mathbb{R}^+\times(\mathbb{R}^+\times\Omega)\rightarrow\mathscr{D}$ for each $t\in(0,\infty)$.  The following proposition characterises when $(Z_t)$ will have c\`adl\`ag paths.
\begin{Proposition}\label{cadlagprop}
The quantile process $(Z_t)$ will have c{\`a}dl{\`a}g paths if and only if the map $t\mapsto F(t,y;\bm{\theta})$ is c{\`a}dl{\`a}g for all $y\in\mathbb{R}$, and $Q_{\zeta}(u;\bm{\xi})$ is a continuous quantile function.
\end{Proposition}
The proof is given in Appendix \ref{cadlagproofappendix}. In what follows, we restrict to the construction of c\`adl\`ag quantile processes. Furthermore, the parameter vectors $\bm{\xi}$ and $\bm{\theta}$ may be generalized to be time--dependent deterministic functions or $(\F_t)$-adapted stochastic processes. This extension is useful, for example, in models where higher moments such as skewness and kurtosis are stochastic and their dynamics are captured by the parameters of the constructed quantile process. The quantile process $Z_t = Q_{\zeta}\left(F(t,Y_t;\bm{\theta}(t));\bm{\xi_t}\right)$ with stochastic parameter vector $(\bm{\xi}_t)$ affords extensive variety, for example, in the dynamic risk profiles captured by the quantile process.  
\

The quantile processes produced by the composite map (\ref{processdrivenQP}) are not unique. Different combinations of $Q_{\zeta}(u;\bm{\xi})$ and $F(t,y;\bm{\theta})$ may produce in law the same quantile process $(Z_t)$. To characterize families of quantile processes, we consider the notion of a {\it pivotal quantile process}. We refer to \cite{shao2006mathematical} for a definition of a pivotal quantity. For our purpose, the marginals of the pivot process are governed by a reference distribution. 

\begin{Definition}\label{pivotalqpdefn}
Let $Y_t\sim F_Y(t,y;\bm{\theta}_Y(t))$ where the parameter vector $\bm{\theta}_Y(t)\in\mathbb{R}^d$. Consider another parameter vector $\bm{\widetilde{\theta}}(t)\in\mathbb{R}^{d^\prime}$, $d^\prime\leq d$, and define the pivot process by $\widetilde{Y_t}:=\mathcal{R}(Y_t;\bm{\theta}_Y(t))$ where $\mathcal{R}$ is a Borel function such that $\widetilde{Y_t}\sim F_{\widetilde{Y}}(t,\tilde{y};\bm{\widetilde{\theta}}(t))$, and the ``reference law'' $F_{\widetilde{Y}}$ does not depend on $\bm{\theta}_Y(t)$. The pivot quantile process formulation is given by
\begin{equation}\label{pivotalquantileprocess}
Z_t = Q_{\zeta}\left(F\left(t,\widetilde{Y}_t;\bm{\theta}(t)\right);\bm{\xi
}(t) \right),
\end{equation}
where $F(t,y;\bm{\theta}(t))$ is a distribution belonging to the same family of distributions as $\widetilde{Y}_t$ for $t\in(0,\infty)$ with parameters $\bm{\theta}(t)\in\mathbb{R}^{d^\prime}$.   
\end{Definition}

The pivot process $(\widetilde{Y_t})$ serves as a base or reference process with respect to which the resulting quantile process is anchored and relative properties between the two processes will be made explicit in terms of the model parameters.  Later, we will proceed to introduce a valuation principle, in the setting of dynamic risk measures, based on quantile processes and their induced probability measures. In this context, the pivot formulation provides a reference measure in relation to which the probability measures induced by an arbitrary quantile process can be compared. In this sense, the pivot--induced measure can be interpreted as an objective probability measure while the quantile--induced measure can be regarded as a subjective measure. Such consideration is also useful in model risk management. While a quantile process may be a firm--specific model, the pivot process could be the baseline imposed by a regulatory body with respect to which the firm--specific (internal) model is assessed. Such an analysis can have material impact on the calculation of capital requirements, which rely on the calculation of risk measures. In Appendix \ref{appendixpivotexamples}, we present two simple examples for the construction of quantile process referencing a pivot. 
\

Next we introduce a canonical composite map for the construction of the quantile process driven by Brownian motion. 
\begin{Definition}\label{canonicaldefn-cont}
Let $(W_t)$ denote a one--dimensional, $(\mathscr{F}_t)$--adapted standard Brownian motion, and set $Y_t=W_t$ in Definition \ref{processdrivendefn}.  For $t\in(0,\infty)$, the canonical quantile process is given by 
\begin{equation}\label{canonicalQP}
Z_t=Q_{\zeta}\left(F_W(t,W_t);\bm{\xi}\right),
\end{equation}
where $F_W(t,w)=[1+\erf(w/\sqrt{2t})]/2$. In case that the range of the quantile function be restricted to some $D_{\zeta}\subset\mathbb{R}$, then $Z_t=Q_{\zeta}(F_W(t,\omega);\bm{\xi}):\mathbb{R}^+\times\Omega\rightarrow D_{\zeta}$.  The canonical quantile process is then defined on the state space $(D_{\zeta},\mathscr{B}(D_{\zeta}))$.
\end{Definition}
As an example, let us consider the Tukey family of quantile processes, see \cite{BMPqd1}, with time--dependent parameters $A(t)\in\mathbb{R}$, $B(t)>0$, $g(t)\neq 0$ and $h(t)>0$ for all $t\in(0,\infty)$. Let $Q_{
\zeta}(u;\bm{\xi})$ where $\zeta=T_{gh}$ and  $\bm{\xi}(t)=(A(t),B(t),g(t),h(t))$, then the canonical Tukey--$gh$ quantile diffusion is given by 
\begin{equation}\label{bmgtruelaw}
Z_t=A(t) + \dfrac{B(t)}{g(t)}\left[\expp\left(g(t)\dfrac{W_t}{\sqrt{t}} \right)-1 \right]\expp\left(h(t)\dfrac{W_t^2}{2t} \right), 
\end{equation}
for $t\in(0,\infty)$. This Tukey--$gh$ class of quantile processes allows one to capture any skew--kurtosis range. Such features are often advantageous in loss modelling, e.g., in non--life insurance settings as discussed in \cite{lmomentsgh}.  To achieve the same flexibility by another stochastic process would require a more complex driving process when compared to standard Brownian motion. It is worth emphasizing that the parameters $g(t)$ and $h(t)$ directly control the skewness and, respectively, the kurtosis of the process $(Z_t)$. In applications that will follow later, we will make extensive use of this fact.
\

We conclude this section by constructing a Tukey--$gh$ quantile process $(Z_t)$ driven by an inhomogeneous Ornstein-Uhlenbeck process $(Y_t)$ satisfying
$$\rd Y_t = \theta(t)\left( \mu(t) -Y_t\right)\rd t + \sigma(t)  \rd W_t$$ with $y_0\in\mathbb{R}^{+}$, $\mu(t)\in\mathbb{R}$, $\sigma(t)\in\mathbb{R}$, and the mean--reversion parameter $\theta(t)\in\mathbb{R}$ for all $t\in(0,\infty)$.  The marginal law of the driving process at each time $t\in(0,\infty)$ is given by
\begin{equation}
F_Y(t,y;\bm{\theta}(t)) = \frac{1}{2}\left[1+\erf\left(\dfrac{y-y_0\e^{-\int_0^t\theta(s)\rd s}-\e^{-\int_0^t\theta(s)\rd s}\int_0^t\expp\left(\int_0^s\theta(u)\rd u  \right)\theta(s)\mu(s)\rd s}{ \e^{-\int_0^t\theta(s)\rd s}\sqrt{2\int_0^t\expp\left(2\int_0^s\theta(u)\rd u \right)\sigma^2(s)\rd s}} \right) \right],
\end{equation}
where $\bm{\theta}(t)=(\theta(t),\mu(t),\sigma(t))$, and so the process $U_t=F_Y(t,Y_t;\bm{\theta}(t))$ is given by
\begin{equation}\label{OUuniformtimeinhom}
U_t=\dfrac{1}{2}\left[1+\erf\left(\dfrac{\int_0^t \expp\left(\int_0^s\theta(u)\rd u \right)\sigma(s)\rd W_s}{\sqrt{2\int_0^t\expp\left(2\int_0^s\theta(u)\rd u\right)\sigma^2(s)\rd s}} \right) \right],
\end{equation}
for $t\in(0,\infty)$. The Tukey--$gh$ quantile process $(Z_t)_{t\in(0,\infty)}$ in this case is obtained by setting $Q_{\zeta}(u;\bm{\xi}(t))=Q_{T_{gh}}(u;A(t),B(t),g(t),h(t))$, and so we have
\begin{equation}\begin{split}
Z_t = A(t) + B(t)\left(\dfrac{\expp\left( g(t)X_t\right)-1}{g(t)} \right) \expp\left(h(t)\left(X_t \right)^2 \right),
\end{split}\end{equation}
where $X_t:=\sqrt{2}\erf^{-}(2U_t-1)$. It is straightforward to derive the SDE satisfied by $(Z_t)$.  An example of a quantile process driven by a variance-gamma process is included in Appendix \ref{appendixpivotexamples}.
\section{Stochastic ordering}\label{stochasticorderingsection}
In this section, we show under what conditions quantile processes satisfy the first-- and second--order stochastic dominance property. Stochastic ordering plays an important part in the novel valuation principle proposed in this work. It implies monotonic behaviour between the level of risk as quantified by the quantile process and the price associated to some contract dependent on the risk, e.g., an insurance premium.  It also allows for a ranking of investor risk profiles, as captured by the moments profile of the quantile process.  In the context of Tukey--$gh$ quantile processes, the monotonic relation between risk and some associated price is reflected by plotting the price as a function of the skew and kurtosis parameters. We analyse the canonical Tukey-$gh$ quantile diffusion processes in relation to stochastic ordering via their parameters $g$ and $h$.
\

We begin with recalling the notion of stochastic ordering, as presented in \cite{levy1992stochastic}, adapted to our context of quantile processes.
\begin{Definition}\label{fosddefn}
Recall Definition \ref{processdrivendefn}, where $Q_{{\zeta}_i}(u):[0,1]\rightarrow[\underline{z}_i,\overline{z}_i]\subseteq\mathbb{R}$ are quantile functions and $F_i(t,y):\mathbb{R}^+\times[\underline{y}_i,\overline{y}_i]\rightarrow[0,1]$ are distribution functions, for $i=1,2$. Consider the quantile processes $Z_t^{(i)}=Q_{{\zeta}_i}(F_i(t,Y_t^{(i)}))$ with marginal distributions $F_{Z^{(i)}}(t,z_i)=\mathbb{P}(Z_t^{(i)}\leq z_i)$, for $z_i\in D_{Z^{(i)}}:=[\underline{z}_i,\overline{z}_i]$. We say that $(Z^{(1)}_t)_{t\in(0,\infty)}$ dominates $(Z^{(2)}_t)_{t\in(0,\infty)}$ by first--, or second--order stochastic dominance on $D_{\zeta}:=[z_0(t),\max\{\overline{z}_1,\overline{z}_2 \}]$, for $z_0(t)\in [\min\{\underline{z}_1,\underline{z}_2\},\max\{\overline{z}_1,\overline{z}_2 \})$ with $F_{Z^{(1)}}(t,z_0(t))=F_{Z^{(2)}}(t,z_0(t))$, if and only if for all $t\in(0,\infty)$ the following hold, respectively:
\begin{description}
    \item FOSD: $F_{Z^{(2)}}(t,z)-F_{Z^{(1)}}(t,z)\geq 0$, for all $z\in D_{\zeta}$,
    \item SOSD: $\int_{z_0(t)}^z\left[ F_{Z^{(2)}}(t,x)- F_{Z^{(1)}}(t,x)\right]\rd x\geq 0$, for all $z\in D_{\zeta}$.
\end{description}
In either stochastic dominance criterion, strict inequality is required for at least one $z\in D_{\zeta}$.
\end{Definition}
The relation between risk preferences via utility functions and FOSD and SOSD criterion are given, with proofs, in \cite{hadar1969rules,hanoch1969efficiency,rothschild1970increasing}. In the following, we consider a compact support $D_Y$. In the case of a non--compact support, the results are analogous.
\begin{Proposition}\label{fosdprop}
Consider $(Z^{(i)}_t)_{t\in(0,\infty)}$, $i=1,2$, in Definition \ref{fosddefn}. Assume $Y_t^{(1)}\succsim_{FOSD}Y_t^{(2)}$ on $D_Y:=[y_0(t),\max\{\overline{y}_1,\overline{y}_2\}]$ where $y_0(t):=\{y_0\in [\min\{\underline{y}_1,\underline{y}_2 \},\max\{\overline{y}_1,\overline{y}_2\}]: F_{Y^{(1)}}(t,y_0)=F_{Y^{(2)}}(t,y_0)\}$ for all $t\in(0,\infty)$. It holds $Z_t^{(1)}\succsim_{FOSD}Z_t^{(2)}$ on $D_{\zeta}:=[z_0(t),\max\{\overline{z}_1,\overline{z}_2 \}]$, if and only if $Q_{\zeta_1}\left(F_1\left(t,y\right)\right)\geq Q_{\zeta_2}\left(F_2\left(t,y\right)\right)$ for all $y\in D_Y$ and $t\in(0,\infty)$, where $z_0(t):=\{\max z_0\in \{[\underline{z}_1,\overline{z}_1]\cup[\underline{z}_2,\overline{z}_2] \}: z_0\leq Q_{\zeta_2}(F_2(t,y_0(t)))\hspace{1mm} \textrm{and}\hspace{1mm} F_{Z^{(1)}}(t,z_0)=F_{Z^{(2)}}(t,z_0)\}$ for all $t\in(0,\infty)$. That is, first--order stochastic dominance is preserved under the pair of composite maps. 
\end{Proposition}
\begin{proof}
Since, by assumption, $Y_t^{(1)}\succsim Y_t^{(2)}$ on $D_Y$, by the definition of FOSD it holds that $F_{Y^{(1)}}(t,y)\leq F_{Y^{(2)}}(t,y)$ for all $y\in D_Y$ and $t\in(0,\infty)$ with strict inequality for at least one $y\in D_Y$.  The quantile processes are constructed as $Z_t^{(i)}=Q_{\zeta_i}( F_i(t,Y_t^{(i)}))$ for $i=1,2$, and so we may write $Y_t^{(i)}=Q_i(t,F_{\zeta_i}(Z_t^{(i)}) )$.  It follows that 
\begin{equation*}
\begin{split}
F_{Y^{(i)}}(t,y)=\mathbb{P}\left(Y_t^{(i)}\leq y\right)&=\mathbb{P}\left(Q_i\left(t,F_{\zeta_i}\left(Z_t^{(i)}\right) \right)\leq y \right)
\\
&=\mathbb{P}\left(Z_t^{(i)}\leq Q_{\zeta_i}\left(F_i\left(t,y\right)\right) \right)=F_{Z^{(i)}}\left(t,Q_{\zeta_i}\left(F_i\left(t,y\right)\right)\right)
\end{split}\end{equation*}
for all $y\in D_Y$, $t\in(0,\infty)$ and so, by assumption,
\begin{equation}\label{zdistinequality-y}
    F_{Z^{(1)}}\left(t,Q_{\zeta_1}\left(F_1\left(t,y\right)\right)\right)\leq F_{Z^{(2)}}\left(t,Q_{\zeta_2}\left(F_2\left(t,y\right)\right)\right)
\end{equation}
for all $y\in D_Y$, $t\in(0,\infty)$.  Eq. (\ref{zdistinequality-y}) can be written as $F_{Z^{(1)}}\left(t,z_1(y)\right)\leq F_{Z^{(2)}}\left(t,z_2(y)\right)$, $y\in D_Y$, where $z_i=Q_{\zeta_i}(F_i(t,y))$, $i=1,2$, and $z_1\neq z_2$. However, we need to show that $F_{Z^{(1)}}\left(t,z\right)\leq F_{Z^{(2)}}\left(t,z\right)$ for all $z\in D_{\zeta}$, for some $D_{\zeta}\subseteq [\min\{\underline{z}_1,\underline{z}_2\},\max\{ \overline{z}_1,\overline{z}_2\}]$ to be determined. We consider the following cases: First, we assume that $z_1(y)\geq z_2(y)$.  Since $F_{Z^{(i)}}(t,z)$, $i=1,2$, are increasing functions in $z$ for all $t\in(0,\infty)$, it holds that 
$$F_{Z^{(1)}}(t,z_2(y))\leq F_{Z^{(1)}}(t,z_1(y))\leq F_{Z^{(2)}}(t,z_2(y))\leq F_{Z^{(2)}}(t,z_1(y)),$$ 
so that $F_{Z^{(1)}}(t,z_2(y))\leq F_{Z^{(2)}}(t,z_2(y))$ and $ F_{Z^{(1)}}(t,z_1(y))\leq F_{Z^{(2)}}(t,z_1(y))$ for all $y\in D_Y$, $t\in(0,\infty)$. It then follows, $F_{Z^{(1)}}(t,z)\leq F_{Z^{(2)}}(t,z)$ for all $z\in [\underset{i}{\min}\, Q_{\zeta_i}(F_i(t,y_0(t))),\max\{\overline{z}_1,\overline{z}_2 \}]=[Q_{\zeta_2}(F_2(t,y_0(t))),\max\{\overline{z}_1,\overline{z}_2 \}]$, with strict inequality for at least the values $z=z_i(y^{*})$, $i=1,2$, where $y^{*}\in D_Y$ is any value such that $F_{Y^{(1)}}(t,y^{(*)})<F_{Y^{(2)}}(t,y^{(*)})$. Next we assume $z_1(y)<z_2(y)$ for all $y\in D_Y$, $t\in(0,\infty)$.  Since $F_{Z^{(1)}}\left(t,z_1(y)\right)\leq F_{Z^{(2)}}\left(t,z_2(y)\right)$, there exists some $y^*\in D_Y$ such that $F_{Z^{(1)}}\left(t,z_1(y^*)\right)= F_{Z^{(2)}}\left(t,z_2(y^*)\right)$ and $z_2(y^{*})>z_1(y^{*})$. Then, $$F_{Z^{(1)}}\left(t,z_2(y^*)\right)>F_{Z^{(1)}}\left(t,z_1(y^*)\right)=F_{Z^{(2)}}\left(t,z_2(y^*)\right),$$ i.e., $F_{Z^{(1)}}\left(t,z_2(y^*)\right)>F_{Z^{(2)}}\left(t,z_2(y^*)\right)$ and so it cannot hold that $F_{Z^{(1)}}(t,z)\leq F_{Z^{(2)}}(t,z)$ for all $z\in [Q_{\zeta_2}(F_2(t,y_0(t))),\max\{\overline{z}_1,\overline{z}_2\}]$. Therefore, we conclude that $F_{Z^{(1)}}(t,z)\leq F_{Z^{(2)}}(t,z)$ for all $z\in[Q_{\zeta_2}(F_2(t,y_0(t))),\max\{\overline{z}_1,\overline{z}_2\}]$, $t\in(0,\infty)$, with strict inequality for at least one $z$ in this range, if and only if $Q_{\zeta_1}(F_1(t,y))\geq Q_{\zeta_2}(F_2(t,y))$.  By the increasing property of distribution functions, it follows that for each $t\in(0,\infty)$ there exists at least one $z\in[\min\{\underline{z}_1,\underline{z}_2\},\max\{\overline{z}_1,\overline{z}_2\}]$ such that $z\leq Q_{\zeta_2}(F_2(t,y_0(t)))$ and $F_{Z^{(1)}}(t,z)=F_{Z^{(2)}}(t,z)$, with $F_{Z^{(1)}}(t,x)\leq F_{Z^{(2)}}(t,x)$ for all $x\in[z,Q_{\zeta_2}(F_2(t,y_0(t))]$.  If we define $$z_0(t):=\left\{\max z\in[\min\{\underline{z}_1,\underline{z}_2\},\max\{\overline{z}_1,\overline{z}_2\}] :z\leq Q_{\zeta_2}(F_2(t,y_0(t)))\hspace{1mm}\&\hspace{1mm}F_{Z^{(1)}}(t,z)=F_{Z^{(2)}}(t,z)\right\},$$ then $F_{Z^{(1)}}(t,z)\leq F_{Z^{(2)}}(t,z)$ for all $z\in D_{\zeta}:=[z_0(t),\max\{\overline{z}_1,\overline{z}_2\}]$ with strict inequality for at least one $z\in D_{\zeta}$, and hence $Z_t^{(1)}\succsim_{FOSD}Z_t^{(2)}$ on $D_{\zeta}\subseteq\mathbb{R}$.
\end{proof}
The situation considered next is where the true law of each driving process $(Y^{(i)}_t)_{i=1,2}$, with $Y^{(1)}_t\neq Y^{(2)}_t$ almost surely for all $t\in(0,\infty)$, is applied in the composite map giving rise to the two quantile processes. 
\begin{Corollary}\label{fosdcorr2}
Consider the case where $F_i(t,y_i)=F_{Y^{(i)}}(t,y_i)$ for $i=1,2$, and all $y_i\in [\underline{y_i},\overline{y_i}]$, $t\in(0,\infty)$. Assume $Y_t^{(1)}\succsim_{FOSD}Y_t^{(2)}$ on $D_Y$. It holds that $Z_t^{(1)}\succsim_{FOSD}Z_t^{(2)}$ on $D_{\zeta}$ for all $t\in(0,\infty)$, if and only if $Q_{\zeta_1}(u)\geq Q_{\zeta_2}(u)$ for all $u\in[F_{Y^{(1)}}(t,y_0(t)),1]$ and $t\in(0,\infty)$.
\end{Corollary}
The proof of this corollary is found in Appendix  \ref{sdproofappendixsection}. Now we present the corollary for the case where the quantile processes are driven by processes equal in distribution.
\begin{Corollary}\label{fosdcorr1}
Assume $Y_t^{(1)}\overset{d}{=}Y_t^{(2)}$, i.e., $F_{Y^{(1)}}(t,y)=F_{Y^{(2)}}(t,y)$ for all $y\in [\min\{\underline{y}_1,\underline{y}_2\},\max\{\overline{y}_1,\overline{y}_2\}]$, $t\in(0,\infty)$. It holds $Z_t^{(1)}\succsim_{FOSD}Z_t^{(2)}$ on $\widetilde{D}_Z:=[\tilde{z}_0(t),\max\{\overline{z}_1,\overline{z}_2\}]$, if and only if $Q_{\zeta_1}(F_1(t,y)\geq Q_{\zeta_2}(F_2(t,y))$ for all $y\in \widetilde{D}_Y:=[\tilde{y}_0(t),\max\{\overline{y}_1,\overline{y}_2\}]$ for some $\tilde{y}_0(t)\in[\min\{\underline{y}_1,\underline{y}_2\},\max\{\overline{y}_1,\overline{y}_2\})$ with strict inequality for at least one $y\in\widetilde{D}_Y$, where $\tilde{z}_0(t):=\{\max\, z_0\in \{[\underline{z}_1,\overline{z}_1]\cup[\underline{z}_2,\overline{z}_2] \}: z_0\leq Q_{\zeta_2}(F_2(t,\tilde{y}_0(t)))\hspace{1mm} \textrm{and}\hspace{1mm} F_{Z^{(1)}}(t,z_0)=F_{Z^{(2)}}(t,z_0)\}$ for all $t\in(0,\infty)$. However, in case that $F_i(t,y)=F_{Y^{(i)}}(t,y)$ for $i=1,2$, for all $y\in D_Y$, and $t\in(0,\infty)$, then $Z_t^{(1)}\succsim_{FOSD}Z_t^{(2)}$ on $\overline{D}_Z:=[z_0,\max\{\overline{z}_1,\overline{z}_2\}]$ where $z_0\in[\min\{\underline{z}_1,\underline{z}_2\},\max\{\overline{z}_1,\overline{z}_2\}]$ with $F_{\zeta_1}(z_0)=F_{\zeta_2}(z_0)$ for all $t\in(0,\infty)$ if and only if $Q_{\zeta_1}(u)\geq Q_{\zeta_2}(u)$ for all $u\in[F_{\zeta_1}(z_0),1]$, with strict inequality for at least one u.
\end{Corollary}
\begin{proof}
The proof is similar to that of Proposition \ref{fosdprop}. However, to ensure 
\begin{equation}\label{fosdstrictconditionz}
    F_{Z^{(1)}}(t,z)<F_{Z^{(2)}}(t,z)
\end{equation} 
for at least one $z\in\widetilde{D}_Z$ and all $t\in(0,\infty)$, see the definition of FOSD, $Q_{\zeta_1}(F_1(t,y))>Q_{\zeta_2}(F_2(t,y))$ is required for at least one $y\in\widetilde{D}_Y$.  This follows from the fact that there does not exist some $y\in[\min\{ \underline{y_1},\underline{y_2}\},\max\{\overline{y_1},\overline{y_2} \}]$ such that $F_{Y^{(1)}}(t,y)<F_{Y^{(2)}}(t,y)$ to impose Inequality \ref{fosdstrictconditionz} whenever $Q_{\zeta_1}(F_1(t,y))\geq Q_{\zeta_2}(F_2(t,y))$ without strict inequality for at least one $y\in \widetilde{D}_Y$ holding almost surely. 

In the case where $F_i(t,y)=F_{Y^{(i)}}(t,y)$, $i=1,2$, one can define $u:=F_1(t,y)=F_2(t,y)=F_{Y^{(1)}}(t,y)=F_{Y^{(2)}}(t,y)$ for all $y\in \widetilde{D}_Y$.  The result follows from writing the inequality $Q_{\zeta_1}(F_1(t,y))\geq Q_{\zeta_2}(F_2(t,y))$ as $Q_{\zeta_1}(u)\geq Q_{\zeta_2}(u)$ for all $u\in [F_{Y^{(1)}}(t,\tilde{y}_0(t),1]$, where strict inequality is required for at least one $u$.  If for all $t\in(0,\infty)$ there exists some $u_0(t)\in[F_{Y^{(1)}}(t,\tilde{y}_0(t),1]$ such that $Q_{\zeta_1}(u_0(t))< Q_{\zeta_2}(u_0(t))$ and $Q_{\zeta_1}(u(t))< Q_{\zeta_2}(u(t))$ for all $u\in[0,u_0(t))$, then $Q_{\zeta_1}(u)\geq Q_{\zeta_2}(u)$ holds for all $u\in[u_0(t),1]$. Hence, $Z_t^{(1)}\succsim_{FOSD}Z_t^{(2)}$ on $\overline{D}_Z:=[\overline{z}_0(t),\max\{\overline{z}_1,\overline{z}_2\}]$, where $\overline{z}_0(t):=Q_{\zeta_1}(u_0(t))=Q_{\zeta_2}(u_0(t))$. In the case where $F_i(t,y)=F_{Y^{(i)}}(t,y)$, $i=1,2$, it holds that
\begin{equation}
    F_{Z^{(i)}}(t,z)=\mathbb{P}\left(Z_t^{(i)}\leq z\right) = \mathbb{P}\left(Y_t^{(i)}\leq Q_{Y^{(i)}}\left(t,F_{\zeta_i}(z) \right) \right) = F_{\zeta_i}(z),
\end{equation}
i.e., the quantile process is stationary.  Thus, by the definition of FOSD, $Z_t^{(1)}\succsim_{FOSD}Z_t^{(2)}$ on $\overline{D}_Z:=[z_0,\max\{\overline{z}_1,\overline{z}_2\}]$ if and only if $F_{\zeta_1}(z)\leq F_{\zeta_2}(z)$ for all $z\in\overline{D}_Z$ with strict inequality for at least one $z$ and where $F_{\zeta_1}(z_0)=F_{\zeta_2}(z_0)$. This is equivalent to $Q_{\zeta_1}(u)\geq Q_{\zeta_2}(u)$ for all $u\in [F_{\zeta_1}(z_0),1]$ with strict inequality for at least one $u$.
\end{proof}
We now proceed to show under what condition quantile processes have the second--order stochastic dominance property. The proofs of the following proposition and corollaries are provided in Appendix \ref{sdproofappendixsection}. We us the short-hand notation $\partial_z\equiv\partial/\partial z$.
\begin{Proposition}\label{sosdprop1}
Consider $(Z^{(i)}_t)_{t\in(0,\infty)}$, $i=1,2$, in Definition \ref{fosddefn}. Assume $Y_t^{(1)}\succsim_{SOSD}Y_t^{(2)}$ on $D_Y:=[y_0(t),\max\{\overline{y}_1,\overline{y}_2\}]$ where $y_0(t):=\{y_0\in [\min\{\underline{y}_1,\underline{y}_2 \},\max\{\overline{y}_1,\overline{y}_2\}]:F_{Y^{(1)}}(t,y_0)=F_{Y^{(2)}}(t,y_0)\}$ for all $t\in(0,\infty)$.  Define $D_{\zeta}:=[z_0(t),\max\{\overline{z}_1,\overline{z}_2\}]$ where $z_0(t):=\underset{i}{\min}\,Q_{\zeta_i}(F_i(t,y_0(t))$, for $i=1,2$. It holds that $Z_t^{(1)}\succsim_{SOSD}Z_t^{(2)}$ on $D_{\zeta}$ for all $t\in(0,\infty)$, if any of the following conditions are satisfied for all $z\in D_{\zeta}$ and $t\in(0,\infty)$:
\begin{enumerate}[(i)]
    \item $\partial_z Q_2(t,F_{\zeta_2}(z))\leq 1\leq \partial_z Q_1(t,F_{\zeta_1}(z))$.
    \item $\partial_z Q_2(t,F_{\zeta_2}(z))\geq 1$, $\partial_z Q_1(t,F_{\zeta_1}(z))\geq 1$, and 
        $\dfrac{F_{Z^{(2)}}(t,z)}{F_{Z^{(1)}}(t,z)}\leq \dfrac{\partial_z Q_1(t,F_{\zeta_1}(z))-1}{\partial_z Q_2(t,F_{\zeta_2}(z))-1}$.
    \item $\partial_z Q_2(t,F_{\zeta_2}(z))\leq 1$, $\partial_z Q_1(t,F_{\zeta_1}(z))\leq 1$, and 
        $\dfrac{F_{Z^{(2)}}(t,z)}{F_{Z^{(1)}}(t,z)}\geq \dfrac{1-\partial_z Q_1(t,F_{\zeta_1}(z))}{1-\partial_z Q_2(t,F_{\zeta_2}(z))}$.
\end{enumerate}
\end{Proposition}
In the next corollary the driving processes are distinct, and their respective true law is applied in the quantile composite map.
\begin{Corollary}\label{sosdcorr2}
Consider the case where $F_i(t,y_i)=F_{Y^{(i)}}(t,y_i)$ for $i=1,2$, and all $y_i\in [\underline{y_i},\overline{y_i}]$, $t\in(0,\infty)$. Assume $Y_t^{(1)}\succsim_{SOSD}Y_t^{(2)}$ on $D_Y$. It holds that $Z_t^{{(1)}}\succsim_{SOSD}Z_t^{(2)}$ on $D_{\zeta}$ for all $t\in(0,\infty)$, if any of the following conditions is satisfied for all $z\in D_{\zeta}$, $t\in(0,\infty)$:
\begin{enumerate}[(i)]
    \item $\partial_z Q_{Y^{(2)}}(t,F_{\zeta_2}(z))\leq 1\leq \partial_z Q_{Y^{(1)}}(t,F_{\zeta_1}(z))$.
    \item $\partial_z Q_{Y^{(2)}}(t,F_{\zeta_2}(z))\geq 1$,  $\partial_z Q_{Y^{(1)}}(t,F_{\zeta_1}(z))\geq 1$, and 
        $\dfrac{F_{\zeta_2}(z)}{F_{\zeta_1}(z)}\leq \dfrac{\partial_z Q_{Y^{(1)}}(t,F_{\zeta_1}(z))-1}{\partial_z Q_{Y^{(2)}}(t,F_{\zeta_2}(z))-1}$.
    \item $\partial_z Q_{Y^{(2)}}(t,F_{\zeta_2}(z))\leq 1$, $\partial_z Q_{Y^{{(1)}}}(t,F_{\zeta_1}(z))\leq 1$, and 
        $\dfrac{F_{\zeta_2}(z)}{F_{\zeta_1}(z)}\geq \dfrac{1-\partial_z Q_{Y^{(1)}}(t,F_{\zeta_1}(z))}{1-\partial_z Q_{Y^{(2)}}(t,F_{\zeta_2}(z))}$.
\end{enumerate}
\end{Corollary}
Analogous to the first-order stochastic dominance analysis, we next give conditions for SOSD in the case where the two quantile processes are generated by different composite maps, with drivers that are equal in distribution.
\begin{Corollary}\label{sosdcorr1}
Assume $Y_t^{(1)}\overset{d}{=}Y_t^{(2)}$, that is $F_{Y^{(1)}}(t,y)=F_{Y^{(2)}}(t,y)$ for all $y\in [\min\{\underline{y}_1,\underline{y}_2 \},\max\{\overline{y}_1,\overline{y}_2\}]$ and $t\in(0,\infty)$.  Define $D_{\zeta}:=[z_0(t),\max\{\overline{z}_1,\overline{z}_2\}]$ where $z_0(t):=\{z_0\in[\min\{\underline{z}_1,\underline{z}_2\},\max\{\overline{z}_1,\overline{z}_2\}) : F_{Z^{(1)}}(t,z_0)=F_{Z^{(2)}}(t,z_0)\}$ for all $t\in(0,\infty)$. It holds that $Z_t^{(1)}\succsim_{SOSD}Z_t^{(2)}$ on $D_{\zeta}$ for all $t\in(0,\infty)$ if either of the conditions (i)--(iii) in Proposition \ref{sosdprop1} hold for all $z\in D_{\zeta}$ and $t\in(0,\infty)$, with strict inequality for at least one $z\in D_{\zeta}$. If, however, $F_i(t,y)=F_{Y^{(i)}}(t,y)$ for $i=1,2$, for all $y\in D_Y$ and $t\in(0,\infty)$, then $Z_t^{(1)}\succsim_{SOSD}Z_t^{(2)}$ on $D_{\zeta}$ if either of the conditions (i)--(iii) in Corollary \ref{sosdcorr2} hold for all $z\in D_{\zeta}$ with strict inequality for at least one $z$, and $t\in(0,\infty)$.
\end{Corollary}
We conclude this section with explicit examples of canonical Tukey quantile processes in the context of stochastic ordering.

\begin{Example}\label{tukeyghSD-ex1}
Let $(W_t^{(i)})_{i=1,2}$ be two independent Brownian motions and
consider the canonical Tukey $g$--$h$ quantile processes 
$Z_t^{(i)}=Q_{T_{gh,i}}\left(F_W\left(t,W_t^{(i)}\right);g_i,h_i \right)$ for $i=1,2$, $g_i\in\mathbb{R}\setminus 0$, $h_i\in\mathbb{R}^+$ and all $t\in(0,\infty)$.  As the (distributionally indistinct) true laws of the driving processes (Brownian motions) are used in the composite maps producing such quantile processes, we consider Corollary \ref{fosdcorr1} with $F_i(t,y)=F_{Y^{(i)}}(t,y)$ for all $y\in\mathbb{R}$.  We consider the following cases:
\begin{enumerate}[(i)]
    \item $g_1>g_2$, $h_1=h_2$.  Then $Z_t^{(1)}\succsim_{FOSD} Z_t^{(2)}$ on $D_{T_{gh}}=\mathbb{R}$ as for fixed $h\in\mathbb{R}^+$, there is monotonicity with respect to the projection $g\mapsto Q_{T_{gh}}(u;g,h)$ since $\partial Q_{T_{gh}}(u;g,h)/\partial g \geq 0$ for all $g\in\mathbb{R}\setminus 0$.
    \item $g_1=g_2$, $h_1>h_2$.  Then $Z_t^{(1)}\succsim_{FOSD} Z_t^{(2)}$ on $[Q_{T_{gh,1}}(0.5),1]$ as $Q_{T_{gh,1}}(u)\geq Q_{T_{gh,2}}(u)$ for all $u\in[0.5,1]$ with equality at $u=0.5$.
    \item $g_1>g_2$, $h_1>h_2$.  Then there exists some $0\leq u^{*}<0.5$ such that $Z_t^{(1)}\succsim_{FOSD} Z_t^{(2)}$ on $[Q_{T_{gh,1}}(u^{*}),1]$.  Here, $Q_{T_{gh,1}}(u^{*})= Q_{T_{gh,2}}(u^{*})$ and $Q_{T_{gh,1}}(u)> Q_{T_{gh,2}}(u)$ for all $u\in(u^{*},1]$.  The larger $g_1-g_2$, or the smaller $h_1-h_2$, the closer $u^{*}$ is to 0.  When $u^{*}=0$, $Z_t^{(1)}\succsim_{FOSD} Z_t^{(2)}$ on $\mathbb{R}$.  The rate of change in $u^{*}$ as $g_1-g_2$ increases for fixed $h_1-h_2$ values is illustrated in Figure \ref{g1g2plot}. 
    \item $g_1>g_2$, $h_1<h_2$.  Then there exists some $0.5< u^{*}\leq 1$ such that $Z_t^{(2)}\succsim_{FOSD} Z_t^{(1)}$ on $[Q_{T_{gh,1}}(u^{*}),1]$.  Here, $Q_{T_{gh,1}}(u^{*})= Q_{T_{gh,2}}(u^{*})$ and $Q_{T_{gh,2}}(u)> Q_{T_{gh,1}}(u)$ for all $u\in(u^{*},1]$.  The larger $g_1-g_2$, or the smaller $h_2-h_1$, the closer $u^{*}$ is to 1.  When $u^{*}=1$, $Z_t^{(1)}\succsim_{FOSD} Z_t^{(2)}$ on $\mathbb{R}$. 
\end{enumerate} 
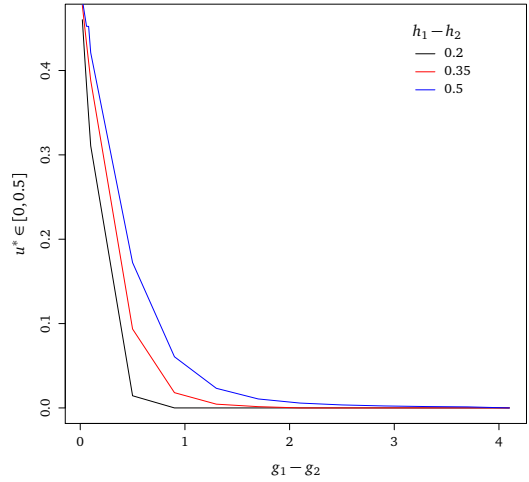
\begin{figure}[H]
\setlength\fwidth{0.8\textwidth}
\centering
\scalebox{0.4}{
\begin{tikzpicture}[x=1pt,y=1pt]
\definecolor{fillColor}{RGB}{255,255,255}
\path[use as bounding box,fill=fillColor,fill opacity=0.00] (0,0) rectangle (505.89,505.89);
\begin{scope}
\path[clip] ( 49.20, 61.20) rectangle (480.69,456.69);
\definecolor{drawColor}{RGB}{0,0,0}

\path[draw=drawColor,line width= 0.4pt,line join=round,line cap=round] ( 65.18,442.04) --
	( 67.14,410.86) --
	( 69.10,380.45) --
	( 71.06,351.10) --
	( 73.01,323.07) --
	(112.18, 87.32) --
	(151.35, 75.85) --
	(190.52, 75.85) --
	(229.69, 75.85) --
	(268.86, 75.85) --
	(308.03, 75.85) --
	(347.20, 75.85) --
	(386.37, 75.85) --
	(425.54, 75.85) --
	(464.71, 75.85);
\end{scope}
\begin{scope}
\path[clip] (  0.00,  0.00) rectangle (505.89,505.89);
\definecolor{drawColor}{RGB}{0,0,0}

\path[draw=drawColor,line width= 0.4pt,line join=round,line cap=round] ( 63.22, 61.20) -- (454.92, 61.20);

\path[draw=drawColor,line width= 0.4pt,line join=round,line cap=round] ( 63.22, 61.20) -- ( 63.22, 55.20);

\path[draw=drawColor,line width= 0.4pt,line join=round,line cap=round] (161.15, 61.20) -- (161.15, 55.20);

\path[draw=drawColor,line width= 0.4pt,line join=round,line cap=round] (259.07, 61.20) -- (259.07, 55.20);

\path[draw=drawColor,line width= 0.4pt,line join=round,line cap=round] (356.99, 61.20) -- (356.99, 55.20);

\path[draw=drawColor,line width= 0.4pt,line join=round,line cap=round] (454.92, 61.20) -- (454.92, 55.20);

\node[text=drawColor,anchor=base,inner sep=0pt, outer sep=0pt, scale=  1.2] at ( 63.22, 39.60) {0};

\node[text=drawColor,anchor=base,inner sep=0pt, outer sep=0pt, scale=  1.2] at (161.15, 39.60) {1};

\node[text=drawColor,anchor=base,inner sep=0pt, outer sep=0pt, scale=  1.2] at (259.07, 39.60) {2};

\node[text=drawColor,anchor=base,inner sep=0pt, outer sep=0pt, scale=  1.2] at (356.99, 39.60) {3};

\node[text=drawColor,anchor=base,inner sep=0pt, outer sep=0pt, scale=  1.2] at (454.92, 39.60) {4};

\path[draw=drawColor,line width= 0.4pt,line join=round,line cap=round] ( 49.20, 75.85) -- ( 49.20,394.11);

\path[draw=drawColor,line width= 0.4pt,line join=round,line cap=round] ( 49.20, 75.85) -- ( 43.20, 75.85);

\path[draw=drawColor,line width= 0.4pt,line join=round,line cap=round] ( 49.20,155.41) -- ( 43.20,155.41);

\path[draw=drawColor,line width= 0.4pt,line join=round,line cap=round] ( 49.20,234.98) -- ( 43.20,234.98);

\path[draw=drawColor,line width= 0.4pt,line join=round,line cap=round] ( 49.20,314.54) -- ( 43.20,314.54);

\path[draw=drawColor,line width= 0.4pt,line join=round,line cap=round] ( 49.20,394.11) -- ( 43.20,394.11);

\node[text=drawColor,rotate= 90.00,anchor=base,inner sep=0pt, outer sep=0pt, scale=  1.2] at ( 34.80, 75.85) {0.0};

\node[text=drawColor,rotate= 90.00,anchor=base,inner sep=0pt, outer sep=0pt, scale=  1.2] at ( 34.80,155.41) {0.1};

\node[text=drawColor,rotate= 90.00,anchor=base,inner sep=0pt, outer sep=0pt, scale=  1.2] at ( 34.80,234.98) {0.2};

\node[text=drawColor,rotate= 90.00,anchor=base,inner sep=0pt, outer sep=0pt, scale=  1.2] at ( 34.80,314.54) {0.3};

\node[text=drawColor,rotate= 90.00,anchor=base,inner sep=0pt, outer sep=0pt, scale=  1.2] at ( 34.80,394.11) {0.4};

\path[draw=drawColor,line width= 0.4pt,line join=round,line cap=round] ( 49.20, 61.20) --
	(480.69, 61.20) --
	(480.69,456.69) --
	( 49.20,456.69) --
	( 49.20, 61.20);
\end{scope}
\begin{scope}
\path[clip] (  0.00,  0.00) rectangle (505.89,505.89);
\definecolor{drawColor}{RGB}{0,0,0}

\node[text=drawColor,anchor=base,inner sep=0pt, outer sep=0pt, scale=  1.4] at (264.94, 15.60) {$g_1-g_2$};

\node[text=drawColor,rotate= 90.00,anchor=base,inner sep=0pt, outer sep=0pt, scale=  1.4] at ( 10.80,258.94) {$u^{*}\in[0,0.5]$};
\end{scope}
\begin{scope}
\path[clip] ( 49.20, 61.20) rectangle (480.69,456.69);
\definecolor{drawColor}{RGB}{255,0,0}

\path[draw=drawColor,line width= 0.4pt,line join=round,line cap=round] ( 65.18,455.58) --
	( 67.14,437.56) --
	( 69.10,419.72) --
	( 71.06,402.11) --
	( 73.01,384.82) --
	(112.18,150.29) --
	(151.35, 90.24) --
	(190.52, 79.40) --
	(229.69, 77.05) --
	(268.86, 75.85) --
	(308.03, 75.85) --
	(347.20, 75.85) --
	(386.37, 75.85) --
	(425.54, 75.85) --
	(464.71, 75.85);
\definecolor{drawColor}{RGB}{0,0,255}

\path[draw=drawColor,line width= 0.4pt,line join=round,line cap=round] ( 65.18,460.99) --
	( 67.14,448.34) --
	( 69.10,435.77) --
	( 71.06,435.77) --
	( 73.01,410.91) --
	(112.18,212.95) --
	(151.35,124.11) --
	(190.52, 94.35) --
	(229.69, 84.32) --
	(268.86, 80.43) --
	(308.03, 78.67) --
	(347.20, 77.72) --
	(386.37, 77.15) --
	(425.54, 76.83) --
	(464.71, 75.85);
\definecolor{drawColor}{RGB}{0,0,0}


\path[draw=drawColor,line width= 0.6pt,line join=round,line cap=round] (377.03,409.89) -- (395.03,409.89);
\definecolor{drawColor}{RGB}{255,0,0}

\path[draw=drawColor,line width= 0.6pt,line join=round,line cap=round] (377.03,392.89) -- (395.03,392.89);
\definecolor{drawColor}{RGB}{0,0,255}

\path[draw=drawColor,line width= 0.6pt,line join=round,line cap=round] (377.03,375.89) -- (395.03,375.89);
\definecolor{drawColor}{RGB}{0,0,0}

\node[text=drawColor,anchor=base,inner sep=0pt, outer sep=0pt, scale=  1.4] at (397.17,426.89) {$h_1-h_2$};

\node[text=drawColor,anchor=base west,inner sep=0pt, outer sep=0pt, scale=  1.2] at (404.03,406.45) {0.2};

\node[text=drawColor,anchor=base west,inner sep=0pt, outer sep=0pt, scale=  1.2] at (404.03,389.45) {0.35};

\node[text=drawColor,anchor=base west,inner sep=0pt, outer sep=0pt, scale=  1.2] at (404.03,372.45) {0.5};
\end{scope}
\end{tikzpicture}}
\caption{Values of $u^{*}$ as $g_1-g_2$ is increased for $g_2=0.1$ fixed, and $h_1-h_2\in\{0.2,0.35,0.5\}$ with $h_2=0.05$ fixed.}
\label{g1g2plot} 
\end{figure}
\end{Example}
Cases (i)---(iv) are illustrated for values of the $g$ and $h$ parameters in Table \ref{SDtable}.
\begin{table}[H]
  \centering
 \begin{tabular}{|l|l|l|l|l|l|}
 \hline
  $g_1$ & $g_2$ & $h_1$ & $h_2$ & $u^{*}$ & Outcome\\ \hline  
  2 & 0.8 & 0.4 & 0.05 & 0.0218 & $Z_t^{(1)}\succsim_{FOSD}Z_t^{(2)}$ on $[-1.109,\infty)$\\
\hline
3 & 0.5 & 0.2 & 0.05 & 0 & $Z_t^{(1)}\succsim_{FOSD}Z_t^{(2)}$ on $\mathbb{R}$\\
\hline
2 & 0.8 & 0.05 & 0.4 & 1 & $Z_t^{(1)}\succsim_{FOSD}Z_t^{(2)}$ on $\mathbb{R}$\\
\hline
2 & 1.5 & 0.05 & 0.2 & 0.985 & $Z_t^{(2)}\succsim_{FOSD}Z_t^{(1)}$ on $[42.36,\infty)$\\
\hline
        \end{tabular}
 \caption{FOSD results for Tukey--$gh$ quantile processes for different values of the skewness and kurtosis parameters.}
  \label{SDtable}
\end{table}

\begin{Example}\label{tukeygSD-ex1}
Consider two canonical Tukey--$g$ quantile processes, %
$Z_t^{(i)}=Q_{T_{g,i}}(F_W(t,W_t^{(i)});g_i)$ for $i=1,2$, %
$g_i\in\mathbb{R}^+\setminus0$ and all $t\in(0,\infty)$. By %
the same argument given in Example \ref{tukeyghSD-ex1}, it %
holds that $Z_t^{(1)}\succsim_{FOSD} Z_t^{(2)}$ if and only %
if $g_1>g_2$.  As such, $Z_t^{(1)}\succsim_{SOSD} Z_t^{(2)}$ 
if and only if $g_1>g_2$. If $g_1< g_2$, then by the same %
argument, $Z_t^{(2)}\succsim_{FOSD}Z_t^{(1)}\implies %
Z_t^{(2)}\succsim_{SOSD}Z_t^{(1)}$, and so 
$Z_t^{(1)}\not\succsim_{SOSD} Z_t^{(2)}$.  As such, when considering the induced measure of a canonical Tukey--$g$ quantile process, we have a risk ordering in the skewness parameter.
\end{Example}
In Appendix \ref{appendixpivotexamples}, we also treat the canonical Tukey--$g$ case where FOSD does not hold, but SOSD is satisfied.

\section{Stochastic valuation principle for financial and insurance risks}\label{valuationprinciplesection}
We recall the filtered probability space $(\Omega, \F,(\F_t)_{t\in[0,\infty)}, \PR)$ whereby $(Y_t)_{t\in[0,T]}$, for $0<T<\infty$, is an $(\mathscr{F}_t)$--adapted c\`adl\`ag process with law $F_Y$. This process may represent a risk, possibly including non--tradable ones, e.g., a climate--related risk. Let us consider $D_Y\subseteq\mathbb{R}$, where $(Y_t)$ is a $D_Y$--valued process.  We assume the filtration contains all information about the development of the risk process and the market.  The goal is to value instruments contingent on $(Y_t)$ in a dynamic, time--consistent manner that accounts for financial market risk (e.g., interest rate risk), investor risk preferences, and characteristics (e.g., skewness, heavy tails) directly associated with the risk process $(Y_t)$.  The choice of the stochastic mappings $\Pi_t(\cdot)$ that characterises the valuation principle will determine how risk factors are reflected in the resultant valuation process and, in an insurance context, the axioms satisfied by the pricing principle.  We define a general dynamic valuation principle as follows.
 \begin{Definition}\label{stochppdefn}
Let $\Pi_t:\mathbb{R}\rightarrow\mathbb{R}$ denote a collection of $(\mathscr{F}_t)$--measurable, continuous and monotonically increasing mappings, or operators, satisfying the properties of positive homogeneity, translational invariance, where $\Pi_t(0)=0$, for all $t\in[0,T]$.  The stochastic valuation principle for the risk $(Y_t)$ is defined as the $(\mathscr{F}_t)$--adapted process $(\Pi_{t,u})_{0\leq t\leq u\leq T}$ where $\Pi_{t,u}:=\Pi_t\left(Y_u \right)$ for all $0\leq t\leq u\leq T$ and such that $\Pi_{s,u}=\Pi_s(\Pi_{t,u})$ for all $0\leq s\leq t\leq u\leq T$.  In the context of a financial or insurance contract, we refer to $(\Pi_{t,u})$ as the stochastic pricing principle for the contract written on the risk $(Y_t)$.  
\end{Definition}
We present some useful definitions that clarify the theory underpinning the following proposition.  First, we give the definition of a convex, conditional risk measure---see e.g., Definition 1 in \cite{acciaio2011dynamic} or \cite{bion2008dynamic}, and Definition 2.3 in \cite{detlefsen2005conditional}.
\begin{Definition}\label{convexconditionalrm}
Consider the filtered probability space and define $L_t^\infty:=L_t^\infty(\Omega,\mathscr{F}_t,\mathbb{P})$ the space of all essentially bounded, $(\mathscr{F}_t)$--measurable random variables, for $t\in[0,\infty)$, and $L_u^\infty:=L_u^\infty(\Omega,\mathscr{F}_u,\mathbb{P})$ for all $0\leq t<u<\infty$.  A map $\varrho_{t,u}(X):=L_u^\infty\rightarrow L_t^\infty$ is called a conditional risk measure if it satisfies the following properties for all $X_u,Y_u\in L_u^\infty$:
\begin{enumerate}[(i)]
    \item {\it Conditional translation invariance:} For all $Y_t\in L_t^\infty$, $\varrho_{ t,u}(Y_u+Y_t)=\varrho_{t,u}(Y_u)-Y_t$;
    \item {\it Monotonicity:} $X_u\leq Y_u\,\Rightarrow\, \varrho_{t,u}(X_u)\geq \varrho_{t,u}(Y_u)$;
    \item {\it Conditional convexity:} For all $\lambda\in L_t^\infty$ where $0\leq\lambda\leq 1$, $\varrho_{t,u}(\lambda Y_u+(1-\lambda)X_u)\leq \lambda\varrho_{t,u}(Y_u)+(1-\lambda)\varrho_{t,u}(X_u)$; 
    \item {\it Normalisation:} $\varrho_{t,u}(0)=0$.
\end{enumerate}
\end{Definition}

One may now define a time--consistent dynamic (in continuous--time), convex risk measure as follows---see e.g., Definitions 5 and 6 in \cite{bion2006dynamic}, Definitions 1 and 2 in \cite{bion2009time}, and Definition 5.1, as well as Remark 5.2, in \cite{detlefsen2005conditional}.
\begin{Definition}\label{dynamicriskmeasuredefn}
A dynamic, convex risk measure on the filtered probability space is a family $(\varrho_{t,u})_{0< t<\infty}$ of conditional convex risk measures for all $0\leq t\leq u<\infty$.  A dynamic risk measure is time--consistent if $\varrho_{s,u}=\varrho_{s,t}(-\varrho_{t,u})$, for all $0\leq s<t<u<\infty$.
\end{Definition}

We may now present the following proposition, that characterises the setting in which the stochastic valuation principle in Definition \ref{stochppdefn} relates to a time--consistent dynamic risk measure. 

\begin{Proposition}
Consider Definition \ref{stochppdefn} and an $(\mathscr{F}_t)$--adapted risk process $(Y_t)$.  If, and only if, $\Pi_t(\cdot)$ is a concave mapping for all $t\in[0,T]$, then $-\Pi_t(Y_u):\mathscr{F}_u\rightarrow\mathscr{F}_t$ is a convex, conditional risk measure, and the family of mappings $(-\Pi_{t,u})_{0\leq t\leq u\leq T}$ is a time--consistent dynamic risk measure.
\end{Proposition}
\begin{proof}
If $\Pi_t(\cdot)$ is a concave mapping for all $t\in[0,T]$, then $-\Pi_t(\cdot)$ will be a convex mapping.  By Definition \ref{stochppdefn}, $-\Pi_t(\cdot)$ is translation invariant and for any real--valued $x\leq y$, $-\Pi_t(y)\leq -\Pi_t(x)$ for all $t\in[0,T]$ as $\Pi_t(y)$ is monotonically increasing.  Thus, under the assumption that $\Pi_t(\cdot)$ is concave, for some $(\mathscr{F}_t)$--adapted risk $(Y_t)$, $-\Pi_t(Y_u)$ adheres to Definition \ref{convexconditionalrm} of a convex, conditional risk measure. It follows, by Definition \ref{dynamicriskmeasuredefn}, that the family of maps $(-\Pi_{t,u})$ is a dynamic risk measure.  The dynamic risk measure $(-\Pi_{t,u})$ is time--consistent, by Definition \ref{dynamicriskmeasuredefn}, if for all $0\leq s\leq t\leq u\leq T$, $-\Pi_{s,u}=-\Pi_{s,t}(\Pi_{t,u})$.  For all $0\leq s\leq t\leq u\leq T$ it holds that $\Pi_{s,t}=\Pi_{s,t}(Y_t)=\Pi_s(Y_t)$
and so, since $\Pi_{s,u}=\Pi_s(\Pi_{t,u})$, it follows that $\Pi_{s,u}=\Pi_{s,t}(\Pi_{t,u})$ and $-\Pi_{s,u}=-\Pi_{s,t}(\Pi_{t,u})$, as required.
\end{proof}

For all $0\leq t\leq u\leq T$, the map $-\Pi_t(\cdot)$ assigns to each random variable $Y_u$ an $\mathscr{F}_t$--measurable random variable $-\Pi_t(Y_u)$ that quantifies the risk given the information at time $t$.  For each $t\in[0,T]$, if the map $\Pi_t(\cdot)$ is linear, it is consistent with the definition of a valuation functional given in \cite{buhlmann1980economic,wuthrich2013financial,wuthrich2010market}.  If $-\Pi_0(\cdot)$ is a sub--additive map, then for each $t\in(0,T]$, $-\Pi_{0,t}$ is a coherent risk measure in the sense of \cite{artzner1999coherent} and thus incorporates static pricing frameworks built by concave distortion operators---see, for example, \cite{godin2012,godin2019,kijima,kijimamuromachi,wang1996,wang,wang2,wang1997axiomatic}.  For $t=0, u=T$, Definition \ref{stochppdefn} covers a large number of premium calculation principles (PCPs), see \cite{laeven2008premium} and the references therein.

\begin{Definition}\label{riskloadingdefn}
Consider Definition \ref{stochppdefn}.  A stochastic pricing principle $(\Pi_{t,u})$ for a risk $(Y_t)$ produces a dynamically consistent risk--loading if $\Pi_{t,u}\geq\mathbb{E}[Y_u\vert\mathscr{F}_t]$ for all $0\leq t\leq u<T$.  If the pricing principle accounts for discounting using the money--market process $(B_t)$, then we require $\Pi_{t,u}\geq B_t\mathbb{E}[Y_u/B_u\vert\mathscr{F}_t]$ for all $0\leq t\leq u<T$.
\end{Definition}
The choice of mappings $\Pi_t(\cdot)$ should ensure that time--consistent prices, or premiums, in line with the treatment of different risks in the literature and the observed behaviour in financial or insurance markets, are produced.  For example, if the framework is applied to the valuation of insurance products, prices ought to be at least higher than the expected value of the loss process to ensure a risk--loading. In the context of a traded financial asset, the valuation framework must be consistent with the principle of no--arbitrage.  In what follows we consider an insurance context whereby markets are not assumed to be complete.  When the underlying risk is not traded, arbitrage--free pricing theory cannot be ensured.

In what follows, we present a special case of the stochastic valuation principle thus introduced, that is characterised by a probability measure. The superscript of the valuation process denotes the measure that it is characterised by.
\begin{Proposition}\label{stochppconditionalex}
Let $V(y):\mathbb{R}\rightarrow\mathbb{R}$ be a Borel--measurable, monotonically increasing function satisfying $V(0)=0$ and $V(y)<\infty$ for all $\vert y\vert<\infty$.  Consider a $\sigma$--finite probability measure on the measurable space $(\Omega,\mathscr{F})$, denoted $\widetilde{\mathbb{P}}$, and define the stochastic operator $\Pi^{\widetilde{\mathbb{P}}}_t(\cdot):=\mathbb{E}^{\widetilde{\mathbb{P}}}[V(\cdot)\vert\mathscr{F}_t]$ for all $t\in[0,\infty)$.
Then for a risk process $(Y_t)$, and by Definition \ref{stochppdefn}, $\Pi_{t,u}^{\widetilde{\mathbb{P}}}:=\Pi^{\widetilde{\mathbb{P}}}_t(Y_u)=\mathbb{E}^{\widetilde{\mathbb{P}}}[V(Y_u)\vert\mathscr{F}_t]$ is a time--consistent stochastic valuation principle for all $0\leq t\leq u\leq T$.
\end{Proposition}
\begin{proof}
By the assumptions made about $V(y)$, and by the properties of the conditional expectation, $\Pi^{\widetilde{\mathbb{P}}}_t(\cdot):=\mathbb{E}^{\widetilde{\mathbb{P}}}[V(\cdot)\vert\mathscr{F}_t]$ is a collection of $(\mathscr{F}_t)$--measurable, continuous and monotonically increasing maps, satisfying the properties in Definition \ref{stochppdefn} for all $t\in[0,T]$. For all $0\leq s\leq t\leq u\leq T$, we have $\Pi_{s,u}^{\widetilde{\mathbb{P}}}=\mathbb{E}^{\widetilde{\mathbb{P}}}[V(Y_u)\vert\mathscr{F}_s]$ and, by the Tower Property of conditional expectation, $\Pi_s^{\widetilde{\mathbb{P}}}(\Pi_{t,u}^{\widetilde{\mathbb{P}}})=\mathbb{E}^{\widetilde{\mathbb{P}}}[\mathbb{E}^{\widetilde{\mathbb{P}}}[V(Y_u)\vert\mathscr{F}_t] \vert\mathscr{F}_s] = \mathbb{E}^{\widetilde{\mathbb{P}}}[V(Y_u)\vert\mathscr{F}_s]=\Pi_{s,u}^{\widetilde{\mathbb{P}}}$.  Therefore $\Pi_{s,u}^{\widetilde{\mathbb{P}}}=\Pi_s(\Pi_{t,u}^{\widetilde{\mathbb{P}}})$ and so the stochastic valuation principle in Proposition \ref{stochppconditionalex} is time--consistent in the sense of Definition \ref{stochppdefn}. 
\end{proof}
\begin{Corollary}
\label{prop1}
Consider an $(\mathscr{F}_t)$--adapted risk $(Y_t)$ and the pricing principle given in Proposition \ref{stochppconditionalex}, that is, $\Pi_{t,u}^{\widetilde{\mathbb{P}}}:=\mathbb{E}^{\widetilde{\mathbb{P}}}[V(Y_u)\vert\mathscr{F}_t]$.  If $\widetilde{\mathbb{P}}$ is such that $\widetilde{\mathbb{P}}(Y_u>y\vert\mathscr{F}_t)\geq\mathbb{P}(Y_u>y\vert\mathscr{F}_t)$ for $y\in D_Y$ and all $0\leq t\leq u\leq T$, then the pricing principle produces a dynamically consistent risk--loading, as per Definition \ref{riskloadingdefn}.  
\end{Corollary}
\begin{proof}
For all $0\leq u\leq t\leq T$, $\Pi_{t,u}^{\widetilde{\mathbb{P}}}=\mathbb{E}^{\widetilde{\mathbb{P}}}[V(Y_u)\vert\mathscr{F}_t]\geq \mathbb{E}^{\widetilde{\mathbb{P}}}[Y_u\vert\mathscr{F}_t]$ since $V(y)$ is an increasing function.  Under the assumption that $\widetilde{\mathbb{P}}(Y_u>y\vert\mathscr{F}_t)\geq\mathbb{P}(Y_u>y\vert\mathscr{F}_t)$ for all $y\in D_Y, 0\leq t\leq u\leq T$, it follows that $\mathbb{E}^{\widetilde{\mathbb{P}}}[Y_u\vert\mathscr{F}_t]\geq\mathbb{E}^{\mathbb{P}}[Y_u\vert\mathscr{F}_t]$ and so $\Pi_{t,u}^{\widetilde{\mathbb{P}}}\geq\mathbb{E}^{\mathbb{P}}[Y_u\vert\mathscr{F}_t]$, as required.
\end{proof}
The stochastic valuation principle given in Proposition \ref{stochppconditionalex} is a general valuation, or pricing, principle that depends not only on the underlying risk $(Y_t)$, but also on the measure $\widetilde{\mathbb{P}}$ which may, for instance, be chosen to incorporate investor preferences or other systemic risk factors, and the function $V(y)$.  If $(\Pi_{t,u}^{\widetilde{\mathbb{P}}})$ in Proposition \ref{stochppconditionalex} produces a risk--loading, $(Y_t)$ becomes more heavy tailed or skewed, or both, under $\widetilde{\mathbb{P}}$, that is, the riskiness of $(Y_t)$ is amplified under $\widetilde{\mathbb{P}}$ relative to $\mathbb{P}$.  We emphasise that the risk loading in Corollary \ref{prop1} can be attributed to both the distorted measure $\widetilde{\mathbb{P}}$ and the monotonically increasing function $V(y)$.  This function may be, for example, a utility function that accounts for the investors risk preferences and how the risk loading built in to the price must accommodate such preferences accordingly.  In expected utility theory, an investor who's preferences satisfy the axioms of transitivity, continuity, and independence will have a von Neumann-Morgenstern utility function that they seek to maximise the expected value of, see \cite{vnm}.  In contradiction with expected utility theory, the function $V(y)$ may also be the value function underlying prospect theory, see \cite{kahneman1979prospect}. Alternatively, one may consider $V(y)$ as the payoff function of a risky contract that, as such, produces a risk loading to compensate an investor for taking on the risk by entering into the contract.  Consider the proof of Corollary \ref{prop1}.  We have
\begin{equation}
\Pi_{t,u}^{\widetilde{\mathbb{P}}}=\mathbb{E}^{\widetilde{\mathbb{P}}}\left[V(Y_u)\vert\mathscr{F}_t\right]\geq \mathbb{E}^{\widetilde{\mathbb{P}}}\left[Y_u\vert\mathscr{F}_t \right]\geq \mathbb{E}^{\mathbb{P}}\left[ Y_u\vert\mathscr{F}_t\right]
\end{equation}
for all $0\leq t<u\leq T$, where the inequality on the left--hand--side is produced by the risk--loading induced by the function $V(y)$, and that on the right hand side is produced by the risk--loading induced by the measure distortion from $\mathbb{P}$ to $\widetilde{\mathbb{P}}$. It is common in many valuation settings to define prices with respect to a pricing kernel (financial mathematics), state price deflator (actuarial mathematics) or state price density and stochastic discount factor (economic theory)---see \cite{buhlmann1980economic,wuthrich2013financial, wuthrich2010market}.  Here, for a given deflator, or pricing kernel, $(\varphi_t)_{t\in[0,\infty)}$ such that $\varphi_t>0$ almost surely for all $t\geq 0$ with $\varphi_0=1$, the valuation process is given by
\begin{equation}\label{pricingkernelpcp}
    \Pi_{t,u}=\frac{1}{\varphi_t}\mathbb{E}^{\mathbb{P}}\left[\varphi_uV(Y_u)\big\vert\mathscr{F}_t\right].
\end{equation}
We do not require that prices should \emph{necessarily} satisfy no--arbitrage unless markets are assumed complete (which in general they are not, in an insurance setting), i.e., $\Pi_{t,u}=V(Y_u)$ should not be implied.  However, this is not saying that trading in such markets \emph{must} lead to arbitrage, i.e., if all market participants do agree on a specific pricing kernel/equivalent martingale measure then we have no--arbitrage prices. However all market participants agreeing on a specific pricing kernel in our setting equates to all market participants preferences corresponding to the same composite map that produces the quantile process from any given driving risk process.

The valuation processes $(\Pi_{t,u})_{0\leq t\leq u\leq T}$ define a pricing system for a given state price deflator $(\varphi_t)$.  In general, there are infinitely many deflators which map from the space of (insurable or financial) risks to corresponding prices.  Such deflators may be determined by factors such as market risk aversion, individual risk preferences, market completeness or incompleteness, legal constraints or tail behaviour of underlying risks, among others.  It is, in general, not assumed that $(\varphi_t)$ be independent of $(Y_t)$.

We note that there is an explicit connection between the pricing kernel and the Radon--Nikodym derivative process that induces the measure change from $\mathbb{P}$ to $\widetilde{\mathbb{P}}$ in Proposition \ref{stochppconditionalex}.  Thus, the task pertains to modelling $(\varphi_t)$ by the Radon--Nikodym derivative process to capture both external risk and risk preferences in the desired way. 
\subsection{Valuation principle based on measure distortions induced by quantile processes}
Now we give a key result of this paper, that is the stochastic pricing principle based on quantile processes.  First, we give definitions of the distorted measures induced by such quantile processes, that is the pushforward measure that defines the law of the process $(Z_t)$.  

Throughout the definitions presented in this section, we refer to Definition \ref{processdrivendefn} for the definition of a quantile process $(Z_t)_{t\in(0,\infty)}$ constructed from some driving process $(Y_t)_{t\in[0,\infty)}$ and where $Q_{\zeta}(u;\bm{\xi}):[0,1]\rightarrow D_{\zeta}\subseteq D_Y\subseteq\mathbb{R}$ so that $(Z_t)$ is an $(\mathscr{F}_t)$--adapted process on the measurable space $(D_{\zeta},\mathscr{B}(D_\zeta))$ where $\mathscr{B}(D_\zeta)\subseteq\mathscr{F}$.
\begin{Definition}\label{qplawdefn}
Let $\mathscr{D}$ denote the collection of all c\`adl\`ag functions from $(0,\infty)$ to $D_\zeta$, that is $t\mapsto Z_t(\omega)$ for all $t\in(0,\infty)$. The quantile process $(Z_t)$ induces a measurable function $Z:\Omega\rightarrow \mathscr{D}$ where $(Z(\omega))(t)=Z_t(\omega)$.  The law of the quantile process is defined to be the pushforward measure $\mathbb{P}^Z(A):=\mathbb{P}(Z^{-}(A))$ for all $A\in\mathscr{B}(D_\zeta)$.
\end{Definition} 
\begin{Definition}\label{distortedmeasuresdefn}
For $n\geq 1$, $0\leq t_1<\ldots<t_n$ and $A\in\mathscr{B}(D_\zeta^n)$, the finite dimensional distributions of the quantile process $(Z_t)$ are defined by the pushforward measures \begin{equation}\mathbb{P}^Z_{t_1,\ldots t_n}(A):=\mathbb{P}\{(Z_{t_1},\ldots,Z_{t_n})\in A\}=\int_{\{\omega\in\Omega\,:\,(Z_{t_1}(\omega),\ldots,Z_{t_n}(\omega))\in A \}}\rd\mathbb{P}(\omega).
\end{equation}
Thus, for all $t\in(0,\infty)$, the ``marginal distorted measure'' induced by the quantile process $(Z_t)$ is defined as the finite dimensional distribution $\mathbb{P}_t^Z(B)$ for all $B\in\mathscr{F}_t$.  
For $0<s<t<\infty$, the ``conditional distorted measure'' is defined as the restriction of $\mathbb{P}^Z_t$ to the sub--$\sigma$--algebra $\mathscr{F}_s$.  That is \begin{equation}\mathbb{P}^Z_{t\vert s}(B):=\mathbb{P}_t^Z(B\vert\mathscr{F}_s)= \mathbb{P}\{Z_t\in B\vert\mathscr{F}_s \}=\int_{\{\omega\in\Omega\,:\, Z_t(\omega) \in B\}} \rd \mathbb{P}(\omega\vert \mathscr{F}_s)\end{equation}
for all $B\in\mathscr{F}_t$.  It holds that $\mathbb{P}^Z_t=\mathbb{P}^Z_{t\vert 0}$. By construction, if $B\in\mathscr{F}_t$ but $B\not\subseteq D_{\zeta}$, then $\mathbb{P}^Z_{t\vert s}(B)=0$ for all $0\leq s<t<\infty$.
\end{Definition}
Assume, now, we wish to incorporate discounting in the stochastic valuation principle where $(B_t)_{t\in[0,\infty)}$ is the money--market process offering a positive rate of return.  The valuation principle based on quantile processes is given as follows.  
\begin{Definition}\label{qppricingprincipledefn}
Recall Proposition \ref{stochppconditionalex} and let $\widetilde{\mathbb{P}}=\mathbb{P}^Z$ be the law of the quantile process, given by Definition \ref{qplawdefn}.  The process $(\Pi_{t,u}^{\mathbb{P}^Z,\,\zeta})_{0\leq t\leq u\leq T}$, defined by 
\begin{equation}\label{Pzprice}
    \Pi_{t,u}^{\mathbb{P}^Z,\,\zeta}=B_t\mathbb{E}^{\mathbb{P}^Z}\left[\frac{1}{B_u}V(Y_u)\vert\mathscr{F}_t\right],
\end{equation}
is the (discounted) quantile process--based stochastic valuation, or pricing, principle (QPVP or QPPP) for the risk $(Y_t)$ and money--market process $(B_t)$, corresponding to the constructive choice of $(Z_t)$.  The second argument in the superscript of the valuation principle process denotes the random variable $\zeta$ that characterises $(Z_t)$---see Definition \ref{processdrivendefn}.  
\end{Definition}
Since $(B_t)$ is the money--market process, Eq. (\ref{Pzprice}) is akin to the quantile process--induced measure $\mathbb{P}^Z$ being a risk--neutral measure.  We note that as markets are not assumed complete, there may exist infinitely many risk--neutral measures, i.e., there is a risk--neutral measure corresponding to each construction of the quantile process $(Z_t)$ that induces the measure $\mathbb{P}^Z$.  We highlight that the transformation from $\mathbb{P}$ to each risk--neutral measure $\mathbb{P}^Z$ involves more than just a drift transformation; the distorted measure is constructed so to account for risk associated to higher--order moments (e.g., skewness, kurtosis), and as such the Girsanov theorem (where the risk adjustment is captured by a first--order/ drift correction) does not apply.
\begin{Proposition}\label{qppricingprincipleprop}
Recall Definition \ref{qppricingprincipledefn} and let the c\'{a}dl\'ag risk process $(Y_t)$ be the stochastic driver used to construct a quantile process $(Z_t)$, with distribution function $F\neq F_Y$ in the composite map.  If $Z_u\succsim_{FOSD}Y_u$, conditional on the sub--$\sigma$--algebra $\mathscr{F}_t$ for all $0\leq t<u\leq T$, then the valuation principle based on a quantile process produces a dynamically consistent risk--loading, i.e., $\Pi_{t,u}^{\mathbb{P}^Z,\,\zeta}\geq B_t\mathbb{E}^{\mathbb{P}}[Y_u/B_u\vert\mathscr{F}_t]$ for all $0\leq t\leq u\leq T$.

\end{Proposition}
\begin{proof}
By construction, $\mathbb{P}\{Z_t\in B\}=\mathbb{P}\{Y_t\in Q(t,F_{\zeta}(B)) \}=\mathbb{P}\{Y_t\in Z^{-}(B)\}=\mathbb{P}^Z\{Y_t\in B\}$ and $\mathbb{P}\{Z_t\in B\vert\mathscr{F}_s \}=\mathbb{P}\{Y_t\in Z^{-}(B)\vert\mathscr{F}_s\}=\mathbb{P}^Z\{Y_t\in B\vert\mathscr{F}_s \}$ for all $B\in\mathscr{F}_t$, i.e., the marginal and conditional distributions of the driving process under the pushforward measure $\mathbb{P}^Z$ coincide with the distributions of the quantile process under $\mathbb{P}$.  As such, it follows that \begin{equation}
    \begin{split}
        \Pi_{t,u}^{\mathbb{P}^Z,\zeta} &=\mathbb{E}^{\mathbb{P}^Z}\left[\frac{B_t}{B_u}V(Y_u)\vert\mathscr{F}_t\right]= \mathbb{E}^{\mathbb{P}}\left[\frac{B_t}{B_u} V(Z_u)\vert\mathscr{F}_t\right]\geq \mathbb{E}^{\mathbb{P}}\left[\frac{B_t}{B_u}Z_u\vert\mathscr{F}_t\right]\geq \mathbb{E}^{\mathbb{P}}\left[\frac{B_t}{B_u}Y_u\vert\mathscr{F}_t\right]
    \end{split}
\end{equation}
where the last two equalities follow from $V(y)$ being an increasing function, and the definition of FOSD.
\end{proof}

\begin{rem}
The stochastic valuation principle given in Proposition \ref{qppricingprincipleprop} is a special case of that defined in Definition \ref{qppricingprincipledefn}, whereby the risk process $(Y_t)$ is the stochastic driver of the quantile process that induces the measure $\mathbb{P}^Z$.  Here, we restrict to the ``false--law'' construction to ensure that any two $\mathbb{P}$--distributionally distinct risks are not $\mathbb{P}^Z$--distributionally indistinct when used as the quantile process drivers for some $Q_{\zeta}$. Moreover, as noted above, in the case that $(B_t)$ is the numeraire associated with the measure $\PR$, we identify $\PR$ as a risk-neutral measure. The measure $\PR^Z$ can then be interpreted as a distorted risk-neutral measure. 
\end{rem}
Since $D_\zeta\subseteq D_Y$, it holds that $\mathbb{P}^Z\ll\mathbb{P}$, and so for all $t\in(0,\infty)$, $\mathbb{P}^Z_t\ll\mathbb{P}\vert_{\mathscr{F}_t}$ where $\mathbb{P}\vert_{\mathscr{F}_t}$ denotes the restriction of $\mathbb{P}$ to the sub--$\sigma$--algebra $\mathscr{F}_t$, and $\mathbb{P}\vert_{\mathscr{F}_t}(A)=\mathbb{P}(A)$ for all $A\in\mathscr{F}_t$.  Similarly, for all $0\leq s<t<\infty$, $\mathbb{P}_{t\vert s}^Z\ll(\mathbb{P}\vert_{\mathscr{F}_t})\vert_{\mathscr{F}_s}$ and so by the Radon--Nikodym theorem, there exists an $\mathscr{F}_t$--measurable function $\rho_{t\vert s}(\omega):\Omega\rightarrow\mathbb{R}_0^+$ such that 
\begin{equation}\label{conditionalrnderiv}
    \mathbb{P}_{t\vert s}^Z(A)=\int_{A}\rho_{t\vert s}(\omega)\rd\mathbb{P}\vert_{\mathscr{F}_s}(\omega)
\end{equation}
for all $A\in\mathscr{F}_t$.  For all $\omega\in\Omega$, $\rd\mathbb{P}_{t\vert s}^Z(\omega)=\rho_{t\vert s}(\omega)\rd (\mathbb{P}\vert_{\mathscr{F}_t})\vert_{\mathscr{F}_s}(\omega)$.  By construction, it holds that $\mathbb{E}^{\mathbb{P}}[\rho_{t\vert s}\vert\mathscr{F}_s]=1$ for all $0\leq s<t<\infty$, and for an $\mathscr{F}_t$--adapted random variable $Y_t$, we have
\begin{equation}
    \mathbb{E}^{\mathbb{P}^Z}[Y_t\vert\mathscr{F}_s]=\mathbb{E}^{\mathbb{P}_{t\vert s}^Z}[Y_t]=\mathbb{E}^{(\mathbb{P}\vert_{\mathscr{F}_t})\vert_{\mathscr{F}_s}}[\rho_{t\vert s}Y_t]=\mathbb{E}^{\mathbb{P}}[\rho_{t\vert s}Y_t\vert\mathscr{F}_s].
\end{equation}

\begin{Definition}
Recall Definition \ref{processdrivendefn} for a continuous quantile process $(Z_t)$.  The Radon--Nikodym derivative in Eq. (\ref{conditionalrnderiv}) is given by
\begin{equation}
    \rho_{t\vert s}(\omega) = \dfrac{\rd F_Y^{\mathbb{P}}\left(t,Q\left(t,F_{\zeta}\left(Y_t(\omega)\right)\right)\vert\mathscr{F}_s\right)}{\rd F_Y^{\mathbb{P}}\left(t, Y_t(\omega)\vert\mathscr{F}_s \right)}
\end{equation}
for all $0<s<t<\infty$ and all $\omega\in\Omega$.  Assuming the existence of the derivatives $\partial_u Q_{\zeta}(u;\bm{\xi}), \, \partial_z F_{\zeta}(z;\bm{\xi})=f_{\zeta}(z;\bm{\xi})$ and $\partial_y F(t,y;\bm{\theta})=f(t,y;\bm{\theta})$, the Radon--Nikodym derivative is given by
\begin{equation}\label{rnderiv1}
\rho_t(\omega) = \dfrac{f_Y^{\mathbb{P}}\left(t,Q\left(t,F_{\zeta}(Y_t(\omega))\right)\right)}{f\left(t,Q\left(t,F_{\zeta}(Y_t(\omega))\right)\right)}\frac{f_{\zeta}(Y_t(\omega))}{f_Y^{\mathbb{P}}(t,Y_t(\omega))}
\end{equation}
for $0=s<t<\infty$ and for all $\omega \in\Omega$.  The case where $(Z_t)$ is a discrete quantile process is similar and one may consider the ratio of conditional probability mass functions of $(Z_t)$ and $(Y_t)$ under $\mathbb{P}$.
\end{Definition}
There is a natural connection between the pricing kernel representation in Eq. (\ref{pricingkernelpcp}) and the QPVP given in Definition \ref{qppricingprincipledefn} when the pricing kernel is defined by $\varphi_u:=\rho_{u\vert t}\varphi_tB_t/B_u$ for all $0\leq t<u<\infty$.

\begin{Proposition}
Let $(B_t)_{t\in[0,\infty)}$ be the money--market process and consider the pricing kernel associated with the $\mathbb{P}$--measure, defined by $\varphi_u:=\rho_{u\vert t}\varphi_tB_u/B_t$ for all $0\leq t<u<\infty$ with $\varphi_0=1$,  The process $(\varphi_tB_t)_{t\in[0,\infty)}$ is a $\mathbb{P}$--martingale.
\end{Proposition}
\begin{proof}
By construction, $(\varphi_tB_t)_{t\in[0,\infty)}$ is $\mathscr{F}_t$--adapted for all $t\in(0,\infty)$.  Recall that $\mathbb{E}^{\mathbb{P}}[\rho_{t\vert s}\vert\mathscr{F}_s]=1$ for all $0\leq s<t<\infty$ and so $\mathbb{E}^{\mathbb{P}}[\varphi_uB_u\vert\mathscr{F}_t]=\mathbb{E}^{\mathbb{P}}[\rho_{u\vert t}\varphi_tB_t\vert\mathscr{F}_t]=\varphi_tB_t\mathbb{E}^{\mathbb{P}}[\rho_{u\vert t}\vert\mathscr{F}_t]=\varphi_tB_t$.  For $0\leq s<t<u<\infty$,
\begin{equation}
    \begin{split}
        \mathbb{E}^{\mathbb{P}}\left[\varphi_uB_u\vert\mathscr{F}_s \right] &= \mathbb{E}^{\mathbb{P}}\left[ \rho_{u\vert t}\varphi_tB_t\vert\mathscr{F}_s\right] = \mathbb{E}^{\mathbb{P}}\left[\rho_{u\vert t}\rho_{t\vert s}\varphi_sB_s\vert\mathscr{F}_s \right] 
        \\ &= \varphi_sB_s \mathbb{E}^{\mathbb{P}}\left[\mathbb{E}^{\mathbb{P}}\left[\rho_{u\vert t}\rho_{t\vert s}\vert\mathscr{F}_t \right]\vert\mathscr{F}_s \right] = \varphi_sB_s \mathbb{E}^{\mathbb{P}}\left[\rho_{t\vert s}\mathbb{E}^{\mathbb{P}}\left[\rho_{u\vert t}\vert\mathscr{F}_t \right]\vert\mathscr{F}_s \right]
        \\
        &= \varphi_sB_s \mathbb{E}^{\mathbb{P}}\left[\rho_{t\vert s}\vert\mathscr{F}_s \right] = \varphi_sB_s
    \end{split}
\end{equation}
where the third and fourth equalities follow from the Tower property and since $\rho_{t\vert s}$ is $\mathscr{F}_t$--adapted, respectively. 
\end{proof}
\subsection{Properties of the stochastic valuation principle based on quantile processes}

Here, we discuss properties of the valuation principle supported by quantile processes in the context of varying risk profiles and preferences.  We note that distorted measure flows in the quantile space are advantageous in that many applications in risk management require the quantification of risk according to quantile functions, as opposed to densities, such as for VaR and dynamic VaR, expected shortfall (ES), and spectral risk measures and their dynamic equivalents.  The quantile process--induced measure allows one to incorporate and parameterise higher--order risk behaviour, investor risk preferences and auxiliary factors (e.g., systemic risk, economic and market conditions) to the valuation principle in Definition \ref{qppricingprincipledefn}.  Such factors may be determined, for instance, by how investors react to other risky contracts or market risks.  We view the distorted measure as subjective; in an insurance setting, given some risk attitude, the valuation principle captures the price an investor would be willing to pay for a contract written on the risk process $(Y_t)$.  The challenge thus lies in selecting a suitable composite map in Definition \ref{processdrivendefn} to construct the quantile process that induces the distorted measure.  The map must be chosen so that the valuation principle in Definition \ref{qppricingprincipledefn} is appropriate given the situation or market under consideration.  

We emphasise that each investor's preferences correspond to a different, but not necessarily unique, composite map used in the construction of a quantile process; the valuation process will capture the preferences of the investor through the induced measure; and the valuation process may be considered relative to those under no distortion.  We consider the $\mathbb{P}$ measure to be some objective baseline to which the distorted, `subjective' risk--neutral measure $\mathbb{P}^Z$ can be compared.  The subjectivity of $\mathbb{P}^Z$ means each investor determines their own risk--neutral measure used for their valuation of financial, insurance, or other risks.  The notion of a subjective probability measure was considered and axiomatised in \cite{savage1954}.  Previous literature in the context of subjective decision making considers an agent's preferences through a subjective utility function, see \cite{vnm}, or by defining scenario based risk measures for risk management, see \cite{acerbi2002spectral,artzner1999coherent}.
We now present the following result on the ranking of valuation principles under different distorted measures.
\newpage
\begin{Proposition}\label{qpvporderingpropn}
Consider the setting in Definition \ref{fosddefn} whereby $Z_t^{(i)}=Q_{\zeta_i}(F_i(t,Y_t^{(i)}))$ are quantile processes for $i=1,2$.  Let $\mathbb{P}^{Z_i}$ be the pushforward measure induced by each quantile process, as given in Definition \ref{qplawdefn}. Define $\Pi_{t,u}^{\mathbb{P}^{Z_i},\,\zeta}=B_t\mathbb{E}^{\mathbb{P}^{Z_i}}[V(Y_u^{(i)})/B_u\vert\mathscr{F}_t]$, as per Definition \ref{qppricingprincipledefn} where the risk process is taken to be the quantile process driver for $i=1,2$. Then, $\Pi_{t,u}^{\mathbb{P}^{Z_1},\,\zeta_1}\geq \Pi_{t,u}^{\mathbb{P}^{Z_2},\,\zeta_2}$ for all $0\leq t<u<\infty$, if and only if $Z_u^{(1)}\succsim_{FOSD} Z_u^{(2)}$ conditional on $\mathscr{F}_t$.
\end{Proposition}
\begin{proof}
Since the risk process underlying each valuation principle is given by the driver of each quantile process, it holds that $ \Pi_{t,u}^{\mathbb{P}^{Z_i},\,\zeta}=B_t\mathbb{E}^{\mathbb{P}^{Z_i}}[V(Y_u^{(i)})/B_u\vert\mathscr{F}_t ] =B_t \mathbb{E}^{\mathbb{P}}[V(Z_u^{(i)})/B_u\vert\mathscr{F}_t]$ for $i=1,2$---see the proof of Proposition \ref{qppricingprincipleprop}.  It is known, see e.g., \cite{hanoch1969efficiency,levy1992stochastic}, that for any non--decreasing function $V(x)$ such that $V(x)<\infty$ for $\vert x\vert<\infty$ (often considered to be a utility function), that Definition \ref{fosddefn} for FOSD is equivalent to $\mathbb{E}^{\mathbb{P}}[V(Z_t^{(1)})]\geq \mathbb{E}^{\mathbb{P}}[V(Z_t^{(2)})]$ for all $t\in(0,\infty)$.  We consider the conditional distribution functions of the quantile processes: for all $0\leq t<u<\infty$, $Z_u^{(1)}\succsim_{FOSD}Z_u^{(2)}$ conditional on $\mathscr{F}_t$ if $F_{Z^{(1)}}^{\mathbb{P}}(u,z\vert\mathscr{F}_t)-F_{Z^{(2)}}^{\mathbb{P}}(u,z\vert\mathscr{F}_t)\geq 0$ for all $z\in D_{\zeta}$ and, equivalently, $\mathbb{E}^{\mathbb{P}}[V(Z_u^{(1)})\vert\mathscr{F}_t]\geq \mathbb{E}^{\mathbb{P}}[V(Z_t^{(2)})\vert\mathscr{F}_t]$.  By the assumptions in Proposition \ref{stochppconditionalex}, $V(y)$ is a non--decreasing function with $V(y)<\infty$ for all $\vert y\vert<\infty$ and so
\begin{equation}
    \Pi_{t,u}^{\mathbb{P}^{Z_1},\,\zeta}=B_t\mathbb{E}^{\mathbb{P}}\left[\frac{1}{B_u}V\left(Z_u^{(1)}\right)\vert\mathscr{F}_t \right] \geq B_t \mathbb{E}^{\mathbb{P}}\left[\frac{1}{B_u}V\left(Z_u^{(2)}\right)\vert\mathscr{F}_t \right]=\Pi_{t,u}^{\mathbb{P}^{Z_2},\,\zeta}
\end{equation}
if, and only if, $Z_u^{(1)}\succsim_{FOSD}Z_u^{(2)}$ conditional on $\mathscr{F}_t$, as required.
\end{proof}

Stochastic dominance results for quantile processes are given in Section \ref{stochasticorderingsection}.  
Such results can be applied in the context of Proposition \ref{qpvporderingpropn} to produce stochastically ordered valuation principles, where the order is characterised by the quantile process composite map, and thus risk preferences or profiles that are embedded into the map.  We emphasise that stochastic ordering in any of the quantile process parameters implies such parameters capture levels of risk--aversion.   We draw attention to the pivot quantile process given in Definition \ref{pivotalqpdefn}.  When considering valuation principles for different underlying risk factors from the same family, this construction may be advantageous: One may map from each risk factor to a common pivot process, which will be used as the quantile process driver, with distribution under the `reference measure' to which all distorted measures may be compared.  The pivot process may be constructed in such a way to account for all baseline risk factors that are common to all market participants and independent of their preferences.  This setting allows for a `relativised' system of valuation principle processes across market participants. 

The flexibility and generality of the quantile process approach to constructing a valuation principle allows the framework to be applicable to a large number of risk--based situations, not restricted to the pricing of financial or insurance contracts.  One may consider non--monetary risks, for example those in the context of disaster modelling or climate and environmental science.  The composite map also need not necessarily capture purely risk preferences, e.g., risk aversion, and can be used, for instance, to reflect the variability of risks or risk profiles involving external market factors (e.g., economical development, government policy, geopolitical factors, technological developments etc.), or historical behaviours (e.g., levels of historical emissions of an economy or business), among other things.  We illustrate this in the following toy example.

\begin{Example}\label{carbontariffexample}
We consider the QPVP in the context of carbon tariffs, where the goal is to cut global emissions that are not controlled under domestic emissions trading schemes (ETS), see, e.g., \cite{worldbank21}.  The aim of the cost-adjustment (tariff) on carbon--intensive imports is to prevent ``carbon leakage'', that is domestic firms taking production to countries with looser environmental standards, and to ensure a level--playing field for foreign and domestic production. Exporters are also incentivized by such schemes to switch to greener production methods.  It is important, however, for the fairness of trade that the tariff is not used as an instrument that unfairly hits imports from a country reliant on exporting such goods.  Tariffs correspond to the price that would have been paid had the good been produced under domestic carbon pricing rules, however given that not all countries have access to the same levels of ``green production''  methods, an indiscriminate carbon tariff scheme could lead to regional inequality and negatively impact export--led development of nations.  Further details on the impact carbon tariffs may have on vulnerable nations under the scheme recently proposed in the EU for implementation in 2026 are given in \cite{CBAMreportUN}.  We consider an illustrative framework that accounts for the emissions of imported goods through the use of the QPVP.  The idea is to relativise carbon tariffs by adjusting the monitored emissions levels that determine the tariff, in line with a number of production--based and economic factors.  Monitoring absolute levels of emissions involved in the production of exported goods is difficult for reasons including, but not limited to, (greener) technological access and advancements, GDP, inflation, geopolitical factors, historical emissions, and whether exporters have paid a domestic carbon price.  The goal is for the mechanism to prevent tariffs drastically impacting less developed or technologically advanced exporters that are disadvantaged in the area of greener production.

Consider a setting with one importer, $n\in\mathbb{N}$ exporters of goods produced within a given industry sector, and $d\in\mathbb{N}$ factors that must be considered in the relativisation of tariffs charged to each exporter.  For $i=1,\ldots,n$, the c\`adl\`ag process $(Y_t^{(i)})$, equipped with $\mathbb{P}$--law $F_{Y^{(i)}}$, corresponds to the absolute level of CO2 emissions (in tonnes) of the production sector for each exporter through time.  Consider Definition \ref{processdrivendefn}: We aim to produce a set of ordered maps $Q_{\zeta_i}(F_i(t,y;\bm{\theta}_i);\bm{\xi}_i)$ such that if each exporter were to produce the same levels of CO2-emissions, the tariffs paid would be ranked fairly.  The distributional families and the choice of parameters will be determined by the $d$ `relativising' factors that impact the level of sophistication in emissions management of each exporter.  One may start, for instance, by considering $F_i(t,y;\bm{\theta})=F_{Y^{(i)}}(t,y-\gamma_i;\bm{\theta})$ where $\gamma_i\in[0,1]$ quantifies the domestic carbon cost already paid, if any.  The exporter with the highest average amount of carbon price paid per tonne of CO2 has parameter $\gamma_i=1$, and the lowest has $\gamma_i=0$.  The remaining `relativising' factors, e.g., those discussed above, would then determine the parameters of the distributional family characterised by $\zeta$.  If $\zeta_i=T_g$, $\bm{\xi}_i=g_i$ so that $Z_t^{(i)}=(\expp(g_i\sqrt{2}\erf^{-}(2F_{Y^{(i)}}(t,Y_t^{(i)}-\gamma_i;\bm{\theta}_i)-1)-1)/g_i$ for all $i\in1,\ldots,n$, the larger the skewness parameter $g_i$, the larger the tariff that will be paid by the $i^{\rm{th}}$ exporter.  One could, for example, construct a weighted index of all relativising factors that determines the skewness parameter of the composite map corresponding to each exporter.  Assume GBP $C\in\mathbb{R}^+$ is the domestic price per tonne of CO2-emissions.  It follows that the stochastic, relative cost of exporting for the $i^{\rm th}$ exporter, with emission levels modelled by the process $(Y_t^{(i)})$, is given by $\Pi_{t,u}^{\mathbb{P}^{Z_i},\,\zeta}=B_t\mathbb{E}^{\mathbb{P}^{Z_i}}[CY_u^{(i)}/B_u\vert\mathscr{F}_t]$ for each $i\in1,\ldots,n$ and all $0\leq t<u<\infty$.  One may refer to Proposition \ref{qpvporderingpropn} for the conditions under which the prices will be ordered.  
\end{Example}
\section{Quantile process--based valuation with a multivariate risk driver}\label{multiqpsection}
In this section we construct quantile processes driven by a multivariate risk driver.  In the setting of the QPVP, this allows for the incorporation of multiple risk factors into the flexible valuation framework.  For example, one may consider systemic risk factors that will impact prices produced under the valuation principle through their dependence structure with the base risk of interest.  The incorporation of multivariate auxiliary risk factors into a premium principle in agricultural insurance is considered in \cite{zhu2019agricultural} where a multivariate weighted premium principle is derived. The valuation principle presented in this section considers auxiliary risk factors similarly.  The framework is advantageous in that there are two `layers' that allow for flexibility and the incorporation of additional information to reduce the extent of mispricing: the multivariate driving risk and the composite map.  The composite map can be selected to factor in investor preferences, to parameterise extreme events that may not be reflected in historical data, and to fairly relativise prices, see e.g., Example \ref{carbontariffexample}.

In what follows, we utilise a copula, see \cite{nelsencopula}, to capture the dependence structure between each marginal, univariate risk factor.  We will employ an $(\mathscr{F}_t)$--adapted, $m$--dimensional process $(\bm{Y}_t)_{t\in[0,\infty)}$ with c\`adl\`ag paths and where $\bm{Y}_t\sim F_{\bm{Y}}(t,y_1,\ldots,y_m;\bm{\vartheta}):\mathbb{R}^+\times\mathbb{R}^m\rightarrow[0,1]$ is the joint law of the $(Y_t^{(i)})$ for $i=1,\ldots,m$ and all $t\in(0,\infty)$.  By Sklar's theorem, there exists a copula $C_Y:\mathbb{R}^+\times[0,1]^m\rightarrow[0,1]$ such that $F_{\bm{Y}}(t,y_1,\ldots,y_m;\bm{\vartheta})=C_Y(t,F_{Y^{(1)}}(t,y_1;\bm{\theta}_1),\ldots, F_{Y^{(m)}}(t,y_m;\bm{\theta_m});\widetilde{\bm{\theta}})$ for all $t\in(0,\infty)$ and where $F_{Y^{(i)}}(t,y_i;\bm{\theta_i})$ is the law associated to the $i^{\rm th}$ marginal of $(\bm{Y}_t)$.  When each $F_{Y^{(i)}}(t,y_i;\bm{\theta_i})$ is continuous for $i=1,\ldots,m$, $C_Y$ is unique. The copula itself may be static, but we consider a time--inhomogeneous copula for generality.  We distinguish between the notion of a {\it `multidimensional'} and a {\it `multivariate'} quantile process, where {\it `multidimensional'} means a univariate quantile process driven by a multivariate risk process, and {\it `multivariate'} means an $n$--dimensional process for $n\geq2$.  Multidimensional quantile processes will be of interest to the following application, and are defined as follows.

\begin{Definition}\label{multidimqp}
Let $Q_{\zeta}(u;\bm{\xi})$ and $F_j(t,y;\bm{\theta}_i)$ be continuous quantile and distribution functions, respectively, as per Definition \ref{processdrivendefn}, for $j=1,\ldots,m$.  For all $t\in(0,\infty)$, let $C(t,u_1,\ldots,u_m;\widetilde{\bm{\theta}}):\mathbb{R}^+\times[0,1]^m\rightarrow[0,1]$ be an $m$--dimensional copula where $\widetilde{\bm{\theta}}\in\mathbb{R}^{d^{\prime\prime}}$ is a $d^{\prime\prime}$--dimensional vector of parameters. Consider an $m$--dimensional $\mathbb{R}^m$--valued c\`adl\`ag process $(\bm{Y}_t)_{t\in[0,\infty)}$.  At each time $t\in(0,\infty)$ the random--level, multidimensional quantile process is defined by
\begin{equation}\label{multiqp}
    Z_t=Q_{\zeta}\left(C\left(t,F_1\left(t,Y_t^{(1)};\bm{\theta}_1\right),\ldots,F_m\left(t,Y_t^{(m)};\bm{\theta}_m\right) ;\widetilde{\bm{\theta}}\right) ;\bm{\xi}\right).
\end{equation}
That is, $$Z_t=Q_{\zeta}\left(C\left(t,F_1(t,Y^{(1)}(t,\omega);\bm{\theta}_1),\ldots,F_m(t,Y^{(m)}(t,\omega);\bm{\theta}_m);\widetilde{\bm{\theta}}\right);\bm{\xi}\right):\mathbb{R}^+\times(\mathbb{R}^+\times\Omega)\rightarrow[-\infty,\infty],$$ and the map $t\mapsto Z(t,\omega)$ for all $\omega\in\Omega$ and $t\in(0,\infty)$ is $\mathscr{F}_t$--measurable.  
\end{Definition}
The output quantile process in Eq. (\ref{multiqp}) is a univariate quantile process, similar to that in Definition \ref{processdrivendefn}, but driven by  a multivariate base process.  If $C(t,u_1,\ldots,u_m;\widetilde{\bm{\theta}})=C_Y(t,u_1,\ldots,u_m;\widetilde{\bm{\theta}})$ and $F_j(t,y;\bm{\theta}_j)=F_{Y^{(j)}}(t,y;\bm{\theta}_j)$ for all $j=1,\ldots,m$, then $Z_t=Q_{\zeta}(F_{\bm{Y}}(t,Y_t^{(1)},\ldots,Y_t^{(m)};\bm{\vartheta}))$.  Here, we say that $C=C_Y$ is the {\it `implicit'} copula, determined by the joint distribution of the multivariate driving process, and that $(Z_t)$ is the {\it `true joint law'} quantile process.  The distribution function induced by the implicit copula is the Kendall distribution function, defined as follows, and the properties of which are discussed in \cite{genest1993statistical,genest2001multivariate,nelsen2003kendall,nelsen2009kendall}.  

 \begin{Definition}\label{kendalldist}
Consider Definition \ref{multidimqp} and assume $F_i(t,y)=F_{Y^{(i)}}(t,y)$ for $i=1,\ldots,m$ is the law associated to each marginal of the multivariate driving process $(\bm{Y}_t)$ so that $U_t^{(i)}:=F_{Y^{(i)}}(t,Y_t^{(i)})$ is uniformly distributed on $[0,1]$ for $i=1,\ldots,m$.  For all $t\in(0,\infty)$ define the random variable $C_t:=C_Y(t,F_{Y^{(1)}}(t,Y_t^{(1)}),\ldots,F_{Y^{(m)}}(t,Y_t^{(m)}))$.  The distribution function of $C_t$ for $t\in(0,\infty)$ is given by the Kendall distribution function $K_{C_Y}(t,v):\mathbb{R}^+\times[0,1]\rightarrow[0,1]$, where $K_{C_Y}(t,v):=\mathbb{P}(C_t\leq v)$ for all $v\in[0,1]$.  
\end{Definition}

We now recall the notion of stochastic ordering, given by Definition \ref{fosddefn}.  Additionally, we refer to \cite{nelsen2003kendall} for the definition of Kendall stochastic ordering of continuous random vectors $(\bm{Y}_t)$ and $(\bm{X}_t)$ for $t\in(0,\infty)$, that is $\bm{Y}_t\succ_K\bm{X}_t$ if and only if $K_{C_Y}(t,v)\leq K_{C_X}(t,v)$ for all $v\in[0,1]$, with strict inequality for at least one $v$, where $C_Y,\, C_X,\, K_{C_Y},\, K_{C_X}$ are given in Definition \ref{kendalldist}.

\begin{Proposition}
Let $(\bm{Y}_t)$ and $(\bm{X}_t)$ be $m$--dimensional $D$--valued, $D\subseteq\mathbb{R}$, stochastic processes with marginal laws $F_{Y^{(i)}}(t,y)$ and $F_{X^{(i)}}(t,x)$ for $i=1,\ldots,m$, respectively, and joint distribution functions $F_{\bm{Y}}(t,y_1,\ldots,y_m)$ and $F_{\bm{X}}(t,x_1,\ldots,x_m)$, respectively, omitting any vectors of parameters.  Let $Q_{\zeta_1}(u):[0,1]\rightarrow D_{\zeta}\subseteq\mathbb{R}$ and $Q_{\zeta_2}(u):[0,1]\rightarrow D_{\zeta}\subseteq\mathbb{R}$ be quantile functions so that $Z_t^{(1)}:=Q_{\zeta_1}(C_Y(t,F_{Y^{(1)}}(t,Y_t^{(1)}),\ldots,F_{Y^{(m)}}(t,Y_t^{(m)})))$ and  $Z_t^{(2)}:=Q_{\zeta_2}(C_X(t,F_{X^{(1)}}(t,X_t^{(1)}),\ldots,F_{X^{(m)}}(t,X_t^{(m)})))$ are multidimensional quantile processes, given by Definition \ref{multidimqp}, where $C_Y$ and $C_X$ are the implicit copulas determined by $F_{\bm{Y}}$ and $F_{\bm{X}}$, respectively.  Assume $\bm{Y}_t\succ_K\bm{X}_t$ on $D$ for all $t\in(0,\infty)$, then $Z_t^{(1)}\succsim_{FOSD} Z_t^{(2)}$ on $D_{\zeta}$ for all $t\in(0,\infty)$ if and only if $Q_{\zeta_1}(u)\geq Q_{\zeta_2}(u)$ for all $u\in[0,1]$.
\end{Proposition}
\begin{proof}
The quantile processes are constructed as $Z_t^{(1)}=Q_{\zeta_1}(C_t^{(1)})$ and $Z_t^{(2)}=Q_{\zeta_2}(C_t^{(2)})$, where $C_t^{(1)}:=C_Y(t,F_{Y^{(1)}}(t,Y_t^{(1)}),\ldots,F_{Y^{(m)}}(t,Y_t^{(m)}))$ and $C_t^{(2)}:=C_X(t,F_{X^{(1)}}(t,X_t^{(1)}),\ldots,F_{X^{(m)}}(t,X_t^{(m)}))$ for all $t\in(0,\infty)$.  It follows that
\begin{equation}\begin{split}
    K_{C_Y}(t,v)=\mathbb{P}\left(C_t^{(1)}\leq v\right) = \mathbb{P}\left(F_{\zeta_1}\left(Z_t^{(1)}\right)\leq v \right)
    =\mathbb{P}\left(Z_t^{(1)}\leq Q_{\zeta_1}\left(v\right) \right)= F_{Z^{(1)}}\left(t,Q_{\zeta_1}(v)\right)
\end{split}\end{equation}
and $K_{C_X}(t,v)=F_{Z^{(2)}}(t,Q_{\zeta_2}(v))$ for all $v\in[0,1]$ and $t\in(0,\infty)$ with strict inequality for at least one $v\in[0,1]$.  As such, by assumption, $F_{Z^{(1)}}(t,Q_{\zeta_1}(v))\leq   F_{Z^{(2)}}(t,Q_{\zeta_2}(v))$
for all $v\in[0,1],\,t\in(0,\infty)$ and the rest of the proof follows by the same argument in the proof of Proposition \ref{fosdprop}, following Eq. (\ref{zdistinequality-y}).
\end{proof}

Definition \ref{multidimqp} allows for quantile processes to be constructed via the transformation of multiple driving risk factors.  The dependence between each risk process is described by the choice of copula, which determines how the quantile process responds to each risk factor and thus how inter--risk dependence structures influence both the distorted pushforward measure and the evaluation of risks under the QPVP.  The choice of quantile function $Q_{\zeta}$ and distribution functions $F_j$, $j=1,\ldots,m$, are chosen to determine the statistical properties and behaviour of the quantile process, e.g., $\zeta=T_{gh}$ if one wishes to parameterise skewness and kurtosis.  We recall that time--dependent or stochastic parameters may also be considered. Since the quantile process given by Definition \ref{multidimqp} is univariate, it follows that the multidimensional QPVP is given by Definition \ref{qppricingprincipledefn} where $(Y_t)$ is a univariate risk and $\mathbb{P}^Z$ is the law of the multidimensional quantile process, given by Definition \ref{distortedmeasuresdefn}, i.e.,
\begin{equation}  \label{multidimqplaw} \mathbb{P}_t^Z(A):=\mathbb{P}\left\{Z_t \in A \right\}=\int_{\left\{\omega\in\Omega\,:\,Q_{\zeta}\left(C\left(t,F_1\left(t,Y_t^{(1)}(\omega)\right),\ldots,F_m\left(t,Y_t^{(m)}(\omega)\right)\right)\right)\in A\right\}}\rd \mathbb{P}(\omega)
\end{equation}
for all $A\in\mathscr{F}_t$ and $t\in(0,\infty)$.
.  When the {\it `true joint law'} quantile process is considered, for any $z\in D_\zeta$, it holds that $\mathbb{P}^Z_t((\inf D_{\zeta},z])=K_C(t,F_{\zeta}(z))$ for all $t\in(0,\infty)$.
\subsection{Calculation of insurance premiums under the multidimensional QPVP}
Here, we apply the multidimensional QPVP to determine the appropriate premium to charge for some risk process $(Y_t)$, in the context of property and casualty (P\&C) insurance.  The reader may refer to \cite{young2014premium} for a discussion on the properties, and examples, of premium principles considered in the literature.  It is important that the premium cover the expected loss of the risk, as well as some risk loading to account for uncertainty and involved costs.  The risk loading induced by the QPVP can be derived from Corollary \ref{prop1}.  The multidimensional premium QPVP for an excess of loss insurance policy (or layer contract) is defined as follows.

\begin{Definition}\label{layercontractpremium}
Consider the multidimensional quantile process given in Definition \ref{multidimqp} where $(\bm{Y}_t)$ is an $m$--dimensional, positive--valued risk process.  Let $\mathbb{P}^Z$ be the distorted measure induced by $(Z_t)$, given by Equation \ref{multidimqplaw}.  Recall Proposition \ref{stochppconditionalex} and let $\widetilde{\mathbb{P}}=\mathbb{P}^Z$ and $V(y)=(y-a)\mathbb{1}_{\{a\leq y<b \}}+(b-a)\mathbb{1}_{\{y\geq b \}}$ for $0<a<b<\infty$.  The function $V(y)$ defines the payoff of a layer contract covering losses between the predefined attachment point $a$ and exhaustion point $b$, with maximum payout $b-a$.  Then for a univariate risk process $(Y_t)$ and money--market process $(B_t)$, the premium QPVP is given by $\Pi_{t,u}^{\mathbb{P}^Z}:=B_t\mathbb{E}^{\mathbb{P}^Z}[V(Y_u)/B_u \vert\mathscr{F}_t]$ for $0\leq t<u\leq T$.
\end{Definition}

The multidimensional premium QPVP for a limited stop--loss contract can be computed using Definition \ref{layercontractpremium} with $V(y)=\min\{(y-a)^+,b\}$ for $a,b>0$.

In what follows, we consider the risk process for which we wish to determine a premium for to be given by $Y_t=Y_t^{(1)}$ where $\bm{Y}_t=(Y_t^{(1)},\ldots,Y_t^{(m)})$.  Intuitively, the multidimensional QPVP is analogous to that given by Definition \ref{qppricingprincipledefn} in the univariate case, however one has a constructive, and flexible, mechanism for the incorporation of auxiliary risk factors, that is $(Y_t^{(2)},\ldots,Y_t^{(m)})$.  In this context, the Radon--Nikodym derivative in Eq. (\ref{conditionalrnderiv}), where $\mathbb{P}_{t\vert s}^Z$ is the restriction of the measure $\mathbb{P}^Z$ in Eq. (\ref{multidimqplaw}) to the sub--$\sigma$--algebra $\mathscr{F}_s$, is given by
\begin{equation}
    \rho_{t\vert s}(\omega)=\dfrac{\rd K_C\left(t,Q_{\zeta}\left(Y_t(\omega)\right)\vert\mathscr{F}_s\right)}{\rd F_Y^{\mathbb{P}}\left(t,Y_t(\omega)\vert\mathscr{F}_s\right)}
\end{equation}
for all $0<s<t<\infty$ and all $\omega\in\Omega$.
\\
\\
\noindent{\bf\Large Acknowledgments}
\\
\\
The authors acknowledge financial support by UK Research \& Innovation via EPSRC CASE project award 1939295. We thank the AIFMRM Research Seminar, University of Cape Town, South Africa (August 2021) and the Workshop Junior Female Researchers in Probability, Berlin, Germany (October 2021) for an opportunity to present aspects of this work.
\bibliography{main}
\newpage
\noindent{\bf\Large Appendix}
\begin{appendix}
\section{Conditions for c\`adl\`ag quantile processes}\label{cadlagproofappendix}
\subsection{Proof of Proposition \ref{cadlagprop}}

\begin{proof}
In the following proof we drop notational dependence of distribution and quantile functions on the parameters.  Since $(Y_t)$ has c{\`a}dl{\`a}g paths, for $s\in(0,\infty)$ and all $\omega\in\Omega$, $\lim_{t\downarrow s}Y(t,\omega)=Y(s,\omega)$ and the left limit $\lim_{t\uparrow
s} Y(t,\omega)$ exists.  By the definition of a distribution function, the map $t\mapsto F(t,y)$ will be c\`adl\`ag for all $t\in(0,\infty)$.  By the definition of a generalised inverse, see \cite{geninverse}, the quantile function $Q_{\zeta}(u)$ will be c\'agl\'ad (when $F_{\zeta}:=Q_{\zeta}^{-}$ is c\`adl\`ag), and continuous when $F_{\zeta}$ is continuous.  If $F(t,y)$ is a c\`adl\`ag function in both arguments, and $Q_{\zeta}(u)$ is a continuous function, it follows that
    \begin{equation}\label{cadlaglimit1}
\lim_{t\downarrow s}Z(t,\omega)=\lim_{t\downarrow s}Q_{\zeta}(F(t,Y(t,\omega))=Q_{\zeta}(F(s,\lim_{t\downarrow s}Y(t,\omega)))=Q_{\zeta}(F(s,Y(s,\omega)))=Z(s,\omega),
\end{equation}
and the limit
\begin{equation}\label{cadlaglimit2}
\lim_{t\uparrow s}Z(t,\omega)=\lim_{t\uparrow s}Q_{\zeta}(F(t,Y(t,\omega))=Q_{\zeta}(F(s,\lim_{t\uparrow s}Y(t,\omega)))
\end{equation}
exists by the c\`adl\`ag property of that paths of $(Y_t)$, so $Z_t=Z(t,\omega)$ has c\`adl\`ag paths for all $\omega\in\Omega$.  We loosen the restrictions of continuity and show, in turn, that if either $t\mapsto F(t,y)$ is not c\`adl\`ag for all $y\in\mathbb{R}$, or $Q_{\zeta}(u)$ is not continuous, $(Z_t)$ will not have c\`adl\`ag paths and thus the proposition holds by contradiction.

First, assume $Q_{\zeta}(u)$ is continuous and $t\mapsto F(t,y)$ is not a c\`adl\`ag map.  Then for each $\omega\in\Omega$,
\begin{equation}
    \lim_{t\downarrow s}Z(t,\omega) = \lim_{t\downarrow s}Q_{\zeta}\left(F\left(t,Y(t,\omega)\right)\right) = Q_{\zeta}\left(\lim_{t\downarrow s}F\left(t,Y(t,\omega)\right) \right) \neq Z(s,\omega)
\end{equation}
as $F(t,y)$ is not right--continuous and so
\begin{equation}
    \lim_{t\downarrow s}F(t,Y(t,\omega))\neq F(s,\lim_{t\downarrow s}Y(t,\omega))=F(s,Y(s,\omega)).
\end{equation}
It follows that $Z_t=Z(t,\omega)$ will not have right--continuous paths.  Now, assume $t\mapsto F(t,y)$ is c{\`a}dl{\`a}g and $Q_{\zeta}(u)$ is c\'agl\'ad but not continuous.  We have
\begin{equation}
    \lim_{t\downarrow s}Q_{\zeta}\left(F\left(t,Y(t,\omega)\right) \right)\neq Q_{\zeta}\left(\lim_{t\downarrow s}F\left(t,Y(t,\omega)\right) \right) = Q_{\zeta}\left(F\left(s,\lim_{t\downarrow s}Y(t,\omega)\right) \right)
\end{equation}
and so $Z_t=Z(t,\omega)$ can not have right--continuous paths.  It follows that the limit in Eq. (\ref{cadlaglimit1}) holds true, and that in Eq. (\ref{cadlaglimit2}) exists if and only if $F(t,y)$ is c\`adl\`ag in both arguments, and $Q_{\zeta}(u)$ is a continuous quantile function.  We only specify that $t\mapsto F(t,y)$ for all $y\in\mathbb{R}$ is c\`adl\`ag as $y\mapsto F(t,y)$ for all $t\in(0,\infty)$ will be c\`adl\`ag by the definition of a distribution function.
\end{proof}
\section{Examples}\label{appendixpivotexamples}

\subsection{Quantile process supported by a marginally Gaussian pivot driver} 
Consider a diffusion $(Y_t)$ such that at each $t\in(0,\infty)$, the process is normally distributed with some time--dependent mean and standard deviation parameters $\mu_Y(t)\in\mathbb{R}$ and $\sigma_Y(t)\in\mathbb{R}^+$, respectively.  The process defined at each time by $\widetilde{Y}_t=(Y_t-\mu_Y(t))/\sigma_Y(t)$ will be distributed according to standard normal distribution, i.e., $\widetilde{\bm{\theta}}=(0,1)$, and so is a pivotal quantity for $\bm{\theta}_Y(t)=(\mu_Y(t),\sigma_Y(t))$.  If we consider a normal distribution function $F$ with $\bm{\theta}(t)=(m(t),v(t))$ where $m(t),v(t)$ are the mean and variance parameters, respectively, the pivotal formulation of the quantile process construction in Definition \ref{processdrivendefn} is given by 
\begin{equation}
Z_t=Q_{\zeta}\left(F(t,\widetilde{Y}_t);\bm{\xi} \right) = Q_{\zeta}\left(\frac{1}{2}\left[1+\erf\left(\dfrac{\widetilde{Y}_t-m(t)}{\sqrt{2v(t)}} \right) \right] ;\bm{\xi}\right)
\end{equation}

In the Tukey--$g$ transform case with time--dependent parameters $\bm{\xi}(t)=(A(t),g(t),B(t))$, we have \begin{equation}\begin{split}\label{gtransformnormalpivotQP}
Z_t &= A(t) + \frac{B(t)}{g(t)}\left[\expp\left(g(t)\dfrac{\widetilde{Y}_t-m(t)}{\sqrt{v(t)}} \right)-1 \right]
\\
&=A(t) + \frac{B(t)}{g(t)}\left[\expp\left(\dfrac{g(t)}{\sigma_Y(t)\sqrt{v(t)}}\left[Y_t-\left(\mu_Y(t)+\sigma_Y(t)m(t) \right)\right] \right)-1 \right]
\end{split}\end{equation}
for each $t\in(0,\infty)$.  We may also write the quantile process in Eq. (\ref{gtransformnormalpivotQP}) as
\begin{equation}\label{gtransformnormalpivotQP-standardised}
Z_t = A^{*}(t)+ \dfrac{B^{*}(t)}{\widetilde{g}(t)}\left[\expp\left(g^{*}(t)Y_t \right)-1\right]
\end{equation}
where
\begin{align}\label{atilde}
A^{*}(t) &= A(t) + \dfrac{B(t)}{g(t)}\left[\expp\left(-\dfrac{g(t)\left(\mu_Y(t)+\sigma_Y(t)m(t)\right)}{\sigma_Y(t)\sqrt{v(t)}} \right)-1 \right]
\\ \label{btilde}
B^{*}(t) &=\dfrac{B(t)}{\sigma_Y(t)\sqrt{v(t)}}
\expp\left(-\dfrac{g(t)\left(\mu_Y(t)+\sigma_Y(t)m(t)\right)}{\sigma_Y(t)\sqrt{v(t)}} \right)
\\\label{gtilde}
g^{*}(t) &= \dfrac{g(t)}{\sigma_Y(t)\sqrt{v(t)}}.
\end{align}
The quantile process in Eq. (\ref{gtransformnormalpivotQP-standardised}) has the form of a standard Tukey--$g$ transform with time--dependent parameters given in Eq. (\ref{atilde})--(\ref{gtilde}), relative to the base $(Y_t)$.  Thus, by definition of the Tukey--$g$ transform, all skewness introduced by the transform will be relative to the base process $(Y_t)$, giving the parameters direct interpretability.   

The first four standardised moments of $Z_t$ at each $t\in(0,\infty)$ are given, in terms of the parameters of the construction composite map, by
\begin{align*}
\mu_Z &=A(t)+ \dfrac{B(t)}{g(t)}\left[\expp\left(-\dfrac{m(t)g(t)}{\sqrt{v(t)}}+\dfrac{g^2(t)}{2v(t)} \right)-1\right],
\\
\sigma^2_Z &= \dfrac{B^2(t)}{g^2(t)} \left(\expp\left(\dfrac{g^2(t)}{v(t)}\right)-1 \right)\expp\left(-\dfrac{2g(t)m(t)}{\sqrt{v(t)}}+\dfrac{g^2(t)}{v(t)} \right)
\\
\widetilde{\mu}_{3,Z} &= \left(\expp\left(\dfrac{g^2(t)}{v(t)}\right)+2 \right)\sqrt{\expp\left(\dfrac{g^2(t)}{v(t)}\right)-1 }
\\
\widetilde{\mu}_{4,Z} &= \expp\left(4\dfrac{g^2(t)}{v(t)} \right) + 2\expp\left(3\dfrac{g^2(t)}{v(t)}\right) + 3\expp\left(2\dfrac{g^2(t)}{v(t)}\right)-6.
\end{align*}
where $\widetilde{\mu}_{3,Z}$ and $\widetilde{\mu}_{4,Z}$ and the third and fourth central moments, respectively.  The skewness and kurtosis of $Z_t$ are given by $\widetilde{\mu}_{3,Z}/\sigma_Z^3$ and $\widetilde{\mu}_{4,Z}/\sigma_Z^4$, respectively.

\subsection{Quantile process supported by a pivot inhomogeneous Poisson process}
Here, we refer to the homogeneous or inhomogeneous Poisson process $\mathscr{N}$ on the positive real line with intensity parameter $\lambda\in\mathbb{R}^+$ or intensity function $\lambda(x)\in[0,\infty)$ for $x\in[0,\infty)$, respectively, by its associated Poisson random measure $N_t:=N([0,t))=\#\{\mathscr{N}\cap[0,t)\}$ for any $t\in[0,\infty)$.

Consider a one--dimensional inhomogeneous Poisson process. We assume that the Poisson mapping theorem applies so that the process can be transformed into a homogeneous Poisson process, which serves as the pivotal quantity at each $t\in(0,\infty)$.  The Poisson mapping theorem, as given in \cite{kingmanpp}, states that for $\mathscr{N}$ a Poisson process with state space $S$ and $\sigma$--finite mean measure $\mu$,  and $\psi:S\rightarrow S^\prime$ a measurable map where $S^\prime$ is another locally compact, separable metric space with Borel $\sigma$--algebra $\mathscr{B}^{\prime}$, if the measure $\mu^\prime(\cdot)=\mu(\psi^{-}(\cdot))$ is non--atomic, then $\psi(\mathscr{N})=\{\psi(\zeta)\hspace{1mm}:\hspace{1mm}\zeta\in\mathscr{N}\}$ is a Poisson process with state space $S^\prime$ and mean measure $\mu^\prime$.

Let $\mathscr{N}$ be an inhomogeneous Poisson process on $S=[0,\infty)$ with intensity function $\lambda(t)\in[0,\infty)$ for $t\in[0,\infty)$ so that its mean measure is defined by $\mu(B)=\int_B\lambda(x)\rd x$ for any $B\in\mathscr{B}$.  Let $\psi$ be a measurable map, and define by $\widetilde{\mathscr{N}}:=\psi(\mathscr{N})$ the Poisson process on state space $S^\prime$ with mean measure $\mu^{\prime}(B)=\mu(\psi^{-}(B))$ for any $B\in\mathscr{B}^\prime$.  Consider some $\widetilde{\lambda}\in\mathbb{R}^+$. If $\psi(x)=\mu^{-}(\widetilde{\lambda} x)$, the Poisson process $\widetilde{\mathscr{N}}$ will be a homogeneous Poisson process with intensity parameter $\widetilde{\lambda}$, and so is a pivotal quantity for $\lambda(t)$.  Now, if we consider a homogeneous Poisson distribution $F_N$ with intensity parameter $\theta\in\mathbb{R}^+$, $\theta\neq\widetilde{\lambda}$, the pivotal formulation of the discrete quantile process construction in Definition \ref{processdrivendefn} is given by
\begin{equation}\label{discretepivotqp}
Z_t = Q_{\zeta}\left(F_N\left(\widetilde{N}_t ;\theta\right) ;\bm{\xi}\right) = Q_{\zeta}\left(\sum_{k=1}^{\widetilde{N}_t}\dfrac{(\theta t)^k\e^{-\theta t}}{k!} ;\bm{\xi}\right)
\end{equation}
for each $t\in(0,\infty)$ and where $\widetilde{N}_t:=\widetilde{N}([0,t))=\#\{\widetilde{\mathscr{N}}\cap[0,t) \}$.  When $\theta=\tilde{\lambda}$, we refer to the quantile process in Eq. (\ref{discretepivotqp}) as the discrete canonical quantile process.

\subsection{Tukey--$g$ quantile process driven by a variance-gamma process}
Let $(\gamma_t)_{t\geq 0}$ be a gamma process with mean rate $\mu t\in\mathbb{R}$ and variance rate $\sigma^2 t\in\mathbb{R}^+$, that is the process with independent increments such that $\gamma_0=0$ and the random variable $\gamma_t$ has gamma distribution with mean $\mu t$ and variance $\sigma t$. We may also define a scale parameter by $\kappa=\mu/\sigma^2$ that inversely controls the jump size, and a ``standardised growth rate'' parameter, $m=\mu^2/\sigma^2$.  The density function of $\gamma_t$ is given by
\begin{equation*}
f_{\gamma}(t,x)=\mathbb{1}_{\{x>0\}}\dfrac{\kappa^{-mt}x^{mt-1}\e^{-x/\kappa}}{\Gamma[mt]}.
\end{equation*}
The gamma process is a pure jump process, and as such may be approximated as a compound Poisson process.  The VG process, first presented in \cite{madan1990variancegamma} and considered in an option pricing setting in \cite{madan1991option, madan1998variancegamma}, can be constructed in the following two ways.  First, as a Brownian motion with drift, subjected to a random time change which follows a gamma process $(\gamma_t^{(1)})_{t\geq0}$ with mean 1 and variance $\nu$, that is
\begin{equation}
Y_t=\widetilde{\mu}\gamma_t^{(1)}+\widetilde{\sigma}W_{\gamma_t^{(1)}}
\end{equation}
for $\widetilde{\mu}\in\mathbb{R}$, $\widetilde{\sigma}\in\mathbb{R}^+$.  Secondly, since the VG process is of finite variation, one may consider it as the difference of two independent gamma processes.  Let $(\gamma_t^{(2)})_{t\geq0}$ and $(\gamma_t^{(3)})_{t\geq0}$ be independent gamma processes with mean parameters $\mu_2,\mu_3$ and variance parameters $\mu_2^2\sigma$ and $\mu_3^2\sigma$, respectively.  
We can define
\begin{equation}
Y_t=\gamma_t^{(2)}-\gamma_t^{(3)}
\end{equation}
where $\mu_2=0.5\sqrt{\widetilde{\mu}^2+2\widetilde{\sigma}^2/\sigma}+0.5\widetilde{\mu}$ and $\mu_3=0.5\sqrt{\widetilde{\mu}^2+2\widetilde{\sigma}^2/\sigma}-0.5\widetilde{\mu}$.  The marginal density of $Y_t$ is given by
\begin{equation}\label{vgmarginaldensity}
f_Y(t,y)=\int_0^\infty \dfrac{1}{\widetilde{\sigma}\sqrt{2\pi u}}\expp\left(-\dfrac{\left(y-\widetilde{\mu}u \right)^2}{2\widetilde{\sigma}^2u} \right)\dfrac{u^{t/\nu-1}\expp\left(-u/\nu \right)}{\nu^{t/\nu}\Gamma[t/\nu]} \rd u
\end{equation}
for $t\in(0,\infty)$.  
We highlight that the density in Eq. (\ref{vgmarginaldensity}) has two more parameters than the normal or lognormal (namely, GBM), which allow for skewness and kurtosis.  
We have
\begin{equation}
F_Y(t,y)=\int_{-\infty}^y f_Y(t,x)\rd x 
\end{equation}
so that $Z_t=Q_{\zeta}(F_Y(t,Y_t);\bm{\xi})$. In the $g$--transform case with constant skewness parameter $g\neq0$ and $A=0,B=1$, we have
\begin{equation}
Z_t =\frac{1}{g}\left[\expp\left(g\sqrt{2}\erf^{-}\left(2\int_{-\infty}^{Y_t}\int_0^\infty \dfrac{1}{\widetilde{\sigma}\sqrt{2\pi u}}\expp\left(-\dfrac{\left(y-\widetilde{\mu}u \right)^2}{2\widetilde{\sigma}^2u} \right)\dfrac{u^{t/\nu-1}\expp\left(-\frac{u}{\nu} \right)}{\nu^{t/\nu}\Gamma[\frac{t}{\nu}]} \rd u\hspace{1mm}\rd y-1 \right) \right)-1 \right].
\end{equation}
\subsection{Example: Two canonical Tukey-$g$ quantile processes with SOSD but no FOSD}
Consider the two canonical Tukey--$g$ quantile processes in Example \ref{tukeygSD-ex1} assuming, however, that the skewness parameter is state--dependent.  Let $g_2\in\mathbb{R}^+\setminus 0$ and 
\begin{equation}
g_1=g_1(z)=
    \begin{cases}
    g_1^a>g_2, & z\leq0
    \\
    g_1^b<g_2, & z>0,
    \end{cases}
\end{equation}
for all $t\in(0,\infty)$ and where $g_1^a, g_1^b\in\mathbb{R}^+\setminus0$.  It holds that $Z_t^{(1)}\not\succsim_{FOSD}Z_t^{(2)}$ on $D_{T_g}:=(-1/g_2,\infty)$.  
We have $F_{Z^{(i)}}(t,z)=0.5[1+\erf(\lnn(g_iz+1)/(g_i\sqrt{2}))]$.  It follows that if $g_1^a, g_1^b, g_2$ are such that \begin{equation}\begin{split}\label{tukeygsosdineq}
    \int_{-1/g_2}^0&\left[\erf\left(\frac{\lnn(g_2x+1)}{g_2\sqrt{2}} \right) - \erf\left(\frac{\lnn(g_1^a x+1)}{g_1^a \sqrt{2}} \right) \right]\rd x 
    \\ &\qquad \geq \int_0^\infty \left[\erf\left(\frac{\lnn(g_1^b x+1)}{g_1^b \sqrt{2}} \right)-\erf\left(\frac{\lnn(g_2x+1)}{g_2\sqrt{2}} \right) \right]\rd x,
\end{split}\end{equation}
then, by Definition \ref{fosddefn}, $Z_t^{(1)}\succsim_{SOSD}Z_t^{(2)}$ on $D_{T_g}$. For example, take $g_1^a=0.8, g_1^b=0.2, g_2=0.3$.  The left--hand integral in Eq. (\ref{tukeygsosdineq}) is equal to  0.1341347, and the right--hand is equal to 0.0660684 and the inequality is satisfied so that, here, $Z_t^{(1)}\succsim_{SOSD}Z_t^{(2)}$ on $D_{T_g}$.
\section{Stochastic ordering results}\label{sdproofappendixsection}
\subsection{Proof of Corollary \ref{fosdcorr2}}\label{fosdcorr2proof}
\begin{proof}
Since, by assumption, $Y_t^{(1)}\succsim Y_t^{(2)}$ on $D_Y$, by the definition of FOSD it holds that $F_{Y^{(1)}}(t,y)\leq F_{Y^{(2)}}(t,y)$ for all $y\in D_Y$ and $t\in(0,\infty)$ with strict inequality for at least one $y\in D_Y$.  The quantile processes are constructed as $Z_t^{(i)}=Q_{\zeta_i}( F_{Y^{(i)}}(t,Y_t^{(i)}))$ for $i=1,2$, and so we may write $Y_t^{(i)}=Q_{Y^{(i)}}(t,F_{\zeta_i}(Z_t^{(i)}) )$.  It follows that 
\begin{equation*}
\begin{split}
F_{Y^{(i)}}(t,y)&=\mathbb{P}\left(Y_t^{(i)}\leq y\right)=\mathbb{P}\left(Q_{Y^{(i)}}\left(t,F_{\zeta_i}\left(Z_t^{(i)}\right) \right)\leq y \right)
\\
&=\mathbb{P}\left(Z_t^{(i)}\leq Q_{\zeta_i}\left(F_{Y^{(i)}}\left(t,y\right)\right) \right)=F_{Z^{(i)}}\left(t,Q_{\zeta_i}\left(F_{Y^{(i)}}\left(t,y\right)\right)\right)
\end{split}\end{equation*}
for all $y\in D_Y$, $t\in(0,\infty)$ and so, by assumption,
\begin{equation}\label{zdistinequality-y-2}
    F_{Z^{(1)}}\left(t,Q_{\zeta_1}\left(F_{Y^{(1)}}\left(t,y\right)\right)\right)\leq F_{Z^{(2)}}\left(t,Q_{\zeta_2}\left(F_{Y^{(2)}}\left(t,y\right)\right)\right)
\end{equation}
for all $y\in D_Y$, $t\in(0,\infty)$.  By the same argument as the proof to Proposition \ref{fosdprop}, under the assumption that the inequality \ref{zdistinequality-y-2} holds, $F_{Z^{(1)}}(t,z)\leq F_{Z^{(2)}}(t,z)$ for all $z\in D_{\zeta}$, with strict inequality for at least one $z\in D_{\zeta}$ whenever $Q_{\zeta_1}(F_{Y^{(1)}}(t,y))\geq Q_{\zeta_2}(F_{Y^{(2)}}(t,y))$ for all $y\in D_Y$, $t\in (0,\infty)$. 

Since $u_1(y):=F_{Y^{(1)}}(t,y)\leq F_{Y^{(2)}}(t,y)=:u_2(y)$ for all $y\in D_Y$, $t\in(0,\infty)$, by the increasing property of quantile functions, it holds that
$$Q_{\zeta_1}(u_2(y))\geq Q_{\zeta_1}(u_1(y))\geq Q_{\zeta_2}(u_2(y)) $$
so that $Q_{\zeta_1}(u_2(y))\geq Q_{\zeta_2}(u_2(y))$ and that $$Q_{\zeta_1}(u_1(y))\geq Q_{\zeta_2}(u_2(y))\geq Q_{\zeta_2}(u_1(y)) $$
so that $Q_{\zeta_1}(u_1(y))\geq Q_{\zeta_2}(u_1(y)) $ for all $y\in D_Y$, $t\in(0,\infty)$.  As such, the condition $Q_{\zeta_1}(F_{Y^{(1)}}(t,y))\geq Q_{\zeta_2}(F_{Y^{(2)}}(t,y))$ reduces to $Q_{\zeta_1}(u)\geq Q_{\zeta_2}(u)$ for all $u\in[F_{Y{(1)}}(t,y_0(t)),1]$, as required.
\end{proof}
\subsection{Proof of Proposition \ref{sosdprop1}}
\begin{proof}
Since, by assumption, $Y_t^{(1)}\succsim_{SOSD}Y_t^{(2)}$, by the definition of SOSD it holds that $\int_{y_0(t)}^y[F_{Y^{(2)}}(t,x)-F_{Y^{(1)}}(t,x)]\rd x\geq 0$ for all $y\in D_Y$, $t\in(0,\infty)$ and with strict inequality for at least one $y\in D_Y$.  The quantile processes are constructed as $Z_t^{(i)}=Q_{\zeta_i}( F_i(t,Y_t^{(i)}))$ for $i=1,2$, and so we may write $Y_t^{(i)}=Q_i(t,F_{\zeta_i}(Z_t^{(i)}) )$.  It follows that 
\begin{equation*}\label{Zdistexpression}
\begin{split}
F_{Y^{(i)}}(t,y)&=\mathbb{P}\left(Y_t^{(i)}\leq y\right)=\mathbb{P}\left(Q_i\left(t,F_{\zeta_i}\left(Z_t^{(i)}\right) \right)\leq y \right)
\\
&=\mathbb{P}\left(Z_t^{(i)}\leq Q_{\zeta_i}\left(F_i\left(t,y\right)\right) \right)=F_{Z^{(i)}}\left(t,Q_{\zeta_i}\left(F_i\left(t,y\right)\right)\right)
\end{split}\end{equation*}
for all $y\in D_Y$, $t\in(0,\infty)$ and so, by assumption,
\begin{equation}\label{sosdineq1-z}
    \int_{y_0(t)}^y\left[F_{Z^{(2)}}\left(t,Q_{\zeta_2}\left(F_2\left(t,x\right) \right)\right)-F_{Z^{(1)}}\left(t,Q_{\zeta_1}\left(F_1\left(t,x\right) \right)\right) \right]\rd x\geq 0
\end{equation}
for all $y\in D_Y$, $t\in(0,\infty)$, with strict inequality for at least one $y\in D_Y$.  By making the changes of variables $v_i:=Q_{\zeta_i}(F_i(t,x))$ for $i=1,2$ so that $\rd v_i/\rd x=f_i(t,x)/f_{\zeta_i}(Q_{\zeta_i}(F_i(t,x)))$, Eq. (\ref{sosdineq1-z}) can be rewritten as 
\begin{equation}\begin{split}
&\int_{Q_{\zeta_2}(F_2(t,y_0(t))}^{Q_{\zeta_2}(F_2(t,y))}F_{Z^{(2)}}\left(t,v_2\right)\dfrac{f_{\zeta_2}\left(v_2\right)}{f_2\left(t,Q_2\left(t,F_{\zeta_2}\left(v_2\right) \right) \right)}\rd v_2 
\\
&\qquad\qquad- \int_{Q_{\zeta_1}(F_1(t,y_0(t))}^{Q_{\zeta_1}(F_1(t,y))}F_{Z^{(1)}}\left(t,v_1\right)\dfrac{f_{\zeta_1}\left(v_1\right)}{f_1\left(t,Q_1\left(t,F_{\zeta_1}\left(v_1\right) \right) \right)}\rd v_1\geq 0.
\end{split}\end{equation}

If $Q_{\zeta_1}(F_1(t,y))\geq Q_{\zeta_2}(F_2(t,y))$ for all $y\in D_Y$ and $t\in(0,\infty)$, then by Proposition \ref{fosdprop},  $Z_t^{(1)}\succsim_{FOSD}Z_t^{(2)}$.  We consider the case where $Q_{\zeta_1}(F_1(t,y))\leq Q_{\zeta_2}(F_2(t,y))$ for at least one $y\in D_Y$.  
It holds that 
\begin{equation}
    \begin{split}
&0\leq \int_{Q_{\zeta_2}(F_2(t,y_0(t))}^{Q_{\zeta_2}(F_2(t,y))}F_{Z^{(2)}}\left(t,v_2\right)\dfrac{f_{\zeta_2}\left(v_2\right)}{f_2\left(t,Q_2\left(t,F_{\zeta_2}\left(v_2\right) \right) \right)}\rd v_2 
\\
&\qquad\qquad- \int_{Q_{\zeta_1}(F_1(t,y_0(t))}^{Q_{\zeta_1}(F_1(t,y))}F_{Z^{(1)}}\left(t,v_1\right)\dfrac{f_{\zeta_1}\left(v_1\right)}{f_1\left(t,Q_1\left(t,F_{\zeta_1}\left(v_1\right) \right) \right)}\rd v_1  
\\[10pt]
&\leq \int_{\underset{i}{\min}Q_{\zeta_i}(F_i(t,y_0(t))}^{\underset{i}{\max}Q_{\zeta_i}(F_i(t,y))} \left[F_{Z^{(2)}}(t,v)\dfrac{f_{\zeta_2}\left(v\right)}{f_2\left(t,Q_2\left(t,F_{\zeta_2}\left(v\right) \right) \right)} -F_{Z^{(1)}}(t,v)\dfrac{f_{\zeta_1}(v)}{f_1\left(t,Q_1\left(t,F_{\zeta_1}(v)\right)\right)}    \right]\rd v
\\
&\qquad\qquad +\int_{\underset{i}{\min}Q_{\zeta_i}(F_i(t,y_0(t)))}^{Q_{\zeta_1}(F_1(t,y_0(t)))}F_{Z^{(1)}}(t,v)\dfrac{f_{\zeta_1}(v)}{f_1\left(t,Q_1\left(t,F_{\zeta_1}(v)\right)\right)}\rd v
\\
&\qquad\qquad\qquad+\int_{Q_{\zeta_1}(F_1(t,y))}^{\underset{i}{\max}Q_{\zeta_i}(F_i(t,y))} F_{Z^{(1)}}(t,v)\dfrac{f_{\zeta_1}(v)}{f_1\left(t,Q_1\left(t,F_{\zeta_1}(v)\right)\right)}\rd v  
\\[10pt]
&\leq \int_{\underset{i}{\min}Q_{\zeta_i}(F_i(t,y_0(t))}^{\underset{i}{\max}Q_{\zeta_i}(F_i(t,y))} \left[F_{Z^{(2)}}(t,v)\dfrac{f_{\zeta_2}\left(v\right)}{f_2\left(t,Q_2\left(t,F_{\zeta_2}\left(v\right) \right) \right)} -F_{Z^{(1)}}(t,v)\dfrac{f_{\zeta_1}(v)}{f_1\left(t,Q_1\left(t,F_{\zeta_1}(v)\right)\right)}    \right]\rd v
\end{split}
\end{equation}
where the last inequality follows from the fact that $F_{Z^{(i)}}(t,z),f_{\zeta_i}(z), f_i(t,Q_i(F_{\zeta_i}(z))) \geq 0$ for all $z\in D_{\zeta}$, $t\in(0,\infty)$. For all $z\in D_{\zeta}:=[z_0(t),\max\{\overline{z}_1,\overline{z}_2\}]$, define
\begin{align}\label{Az}
    A(z)&:= \dfrac{f_{\zeta_2}(z)}{f_2\left(t,Q_2\left(t,F_{\zeta_2}(z)\right)\right)}
    \\ \label{Bz}
    B(z)&:= \dfrac{f_{\zeta_1}(z)}{f_1\left(t,Q_1\left(t,F_{\zeta_1}(z)\right)\right)}.
\end{align}
It holds that 
\begin{equation*}
    A(z)F_{Z^{(2)}}(t,z)-B(z)F_{Z^{(1)}}(t,z)\leq F_{Z^{(2)}}(t,z)-F_{Z^{(1)}}(t,z)
\end{equation*}
whenever
\begin{enumerate}[(i)]
    \item $A(z)\leq 1$ and $B(z)\geq 1$,
    \item $A(z),B(z)\geq 1$ and $
       A(z)F_{Z^{(2)}}(t,z)-F_{Z^{(2)}}(t,z) \leq B(z)F_{Z^{(1)}}(t,z)-F_{Z^{(1)}}(t,z)$, \\ i.e., $F_{Z^{(2)}}(t,z)/F_{Z^{(1)}}(t,z)\leq (B(z)-1)/(A(z)-1)$,
     \item $A(z)\geq 1$, $B(z)\leq 1$ and $
       F_{Z^{(2)}}(t,z)-A(z)F_{Z^{(2)}}(t,z) \geq F_{Z^{(1)}}(t,z)-B(z)F_{Z^{(1)}}(t,z)$, \\
       i.e., $F_{Z^{(2)}}(t,z)/F_{Z^{(1)}}(t,z)\geq (1-B(z))/(1-A(z))$
\end{enumerate}
for all $z\in D_{\zeta}$, $t\in(0,\infty)$.  As such, whenever either of conditions (i)--(iii) hold with $A(z),B(z)$ given by Eq. (\ref{Az}) and (\ref{Bz}), respectively, it follows that 
\begin{equation}\label{sosdineq1-z-3}
    \begin{split}
0&\leq\int_{z_0(t)}^{z} \left[F_{Z^{(2)}}(t,v)\dfrac{f_{\zeta_2}\left(v\right)}{f_2\left(t,Q_2\left(t,F_{\zeta_2}\left(v\right) \right) \right)} -F_{Z^{(1)}}(t,v)\dfrac{f_{\zeta_1}(v)}{f_1\left(t,Q_1\left(t,F_{\zeta_1}(v)\right)\right)}    \right]\rd v \\
&\leq \int_{z_0(t)}^{z} \left[F_{Z^{(2)}}(t,v) -F_{Z^{(1)}}(t,v)    \right]\rd v
    \end{split}
\end{equation}
where $z_0(t):=\underset{i}{\min}Q_{\zeta_i}(t,y_0(t))$ for each $t\in(0,\infty)$ and $z:=\underset{i}{\max}Q_{\zeta_i}(F_i(t,y))$ for all $y\in D_Y$, $t\in(0,\infty)$.  If $\max y\in D_Y=:\overline{y}$, then $\underset{i}{\max}Q_{\zeta_i}(F_i(t,\overline{y}))=\max\{\overline{z}_1,\overline{z}_2\}:=\max z\in D_{\zeta}$. Since the inequality (\ref{sosdineq1-z}) is strict for at least one $y\in D_Y$, that in Eq. (\ref{sosdineq1-z-3}) will be strict for at least one $z\in D_{\zeta}$, corresponding to the quantile transformation of such $y$ value.  Here, if the above conditions are satisfied, $Z_t^{(1)}\succsim_{SOSD}Z_t^{(2)}$ on $D_{\zeta}:=[z_0(t),\max\{\overline{z}_1,\overline{z}_2\}]$ by the definition of SOSD, as required.
\end{proof}
\subsection{Proof of Corollary \ref{sosdcorr2}}
\begin{proof}
The proof is analogous to that of Proposition \ref{sosdprop1}, however it now holds that
\begin{equation}
    \begin{split}
        F_{Z^{(i)}}(t,z) &= \mathbb{P}\left(Z_t^{(i)}\leq z \right) = \mathbb{P}\left(Q_{\zeta_i}\left(F_{Y^{(i)}}(t,Y_t^{(i)} \right)\leq z \right)
        \\
         &= \mathbb{P} \left(Y_t^{(i)}\leq Q_{Y^{(i)}}\left(t,F_{\zeta_i}(z)\right) \right) = F_{Y^{(i)}}\left(t,Q_{Y^{(i)}}\left(t,F_{\zeta_i}(z)\right)\right) = F_{\zeta_i}(z)
    \end{split}
\end{equation}
for $i=1,2$ and all $z\in[\min\{\underline{z}_1,\underline{z}_2 \},\max\{\overline{z}_1,\overline{z}_2 \}]$, $t\in(0,\infty)$, and so we replace $F_{Z^{(i)}}(t,z)$ and $Q_{i}(t,u)$ with $F_{\zeta_i}(z)$ and $Q_{Y^{(i)}}(t,u)$, respectively, in conditions (i)--(iii) in the proposition.
\end{proof}
\subsection{Proof of Corollary \ref{sosdcorr1}}
\begin{proof}
The proof is similar to that of Proposition \ref{sosdprop1}, however it holds that $F_{Y^{(1)}}(t,y)=F_{Y^{(2)}}(t,y)$ for all $y\in[\min\{\underline{y_1},\underline{y_2}\},\max\{\overline{y_1},\overline{y_2}\}]$, $t\in(0,\infty)$ and so $\int_{\min\{\underline{y_1},\underline{y_2}\}}^y [F_{Y^{(2)}}(t,x)-F_{Y^{(1)}}(t,x)]\rd x=0$ for all $y\in[\min\{\underline{y_1},\underline{y_2}\},\max\{\overline{y_1},\overline{y_2}\}]$, $t\in(0,\infty)$.  By Eq. (\ref{Zdistexpression}), this can equivalently be written as
\begin{equation}\label{zsosdcondition3}
   \int_{y_0(t)}^y\left[F_{Z^{(2)}}(t,Q_{\zeta_2}\left(F_2(t,x) \right) - F_{Z^{(1)}}(t,Q_{\zeta_1}\left(F_1(t,x) \right) \right]\rd x=0
\end{equation}
for all $y\in[\min\{\underline{y_1},\underline{y_2}\},\max\{\overline{y_1},\overline{y_2}\}]$, $y_0(t)\in[\min\{\underline{y_1},\underline{y_2}\},\max\{\overline{y_1},\overline{y_2}\}\})$ and $t\in(0,\infty)$.  Consider the case where $Q_{\zeta_1}(F_1(t,y))\leq Q_{\zeta_2}(F_2(t,y))$ for at least one $y\in [\min\{\underline{y_1},\underline{y_2}\},\max\{\overline{y_1},\overline{y_2}\}]$ for all $t\in(0,\infty)$ so that, by Proposition \ref{fosdprop} $Z_t^{(1)}\not\succsim_{FOSD}Z_t^{(2)}$. Making the changes of variables $v_i:=Q_{\zeta_i}(F_i(t,x))$ for $i=1,2$, Eq. (\ref{zsosdcondition3}) can be rewritten as 
\begin{equation}
    \begin{split}
     0 &= \int_{Q_{\zeta_2}(F_2(t,y_0(t)))}^{Q_{\zeta_2}(F_2(t,y))} \dfrac{f_{\zeta_2}(v_2)}{f_2\left(t,Q_2\left(t,F_{\zeta_2}(v_2) \right) \right)} F_{Z^{(2)}}(t,v_2)   \rd v_2 
     \\
     &\qquad- \int_{Q_{\zeta_1}(F_1(t,y_0(t)))}^{Q_{\zeta_1}(F_1(t,y))} \dfrac{f_{\zeta_1}(v_1)}{f_1\left(t,Q_1\left(t,F_{\zeta_1}(v_1) \right) \right)} F_{Z^{(1)}}(t,v_1)   \rd v_1
     \\[10pt]
     &\leq \int_{\underset{i}{\min}Q_{\zeta_i}(F_i(t,y_0(t)))}^{\underset{i}{\max}Q_{\zeta_i}(F_i(t,y))} \left[\dfrac{f_{\zeta_2}(v)}{f_2\left(t,Q_2\left(t,F_{\zeta_2}(v) \right) \right)} F_{Z^{(2)}}(t,v) - \dfrac{f_{\zeta_1}(v)}{f_1\left(t,Q_1\left(t,F_{\zeta_1}(v) \right) \right)} F_{Z^{(1)}}(t,v)   \right]\rd v
     \\
     &\qquad+ \int_{Q_{\zeta_1}(F_1(t,y))}^{\underset{i}{\max}Q_{\zeta_i}(F_i(t,y))} \dfrac{f_{\zeta_1}(v)}{f_1\left(t,Q_1\left(t,F_{\zeta_1}(v) \right) \right)} F_{Z^{(1)}}(t,v)   \rd v
     \\
     &\qquad+\int_{\underset{i}{\min}Q_{\zeta_i}(F_i(t,y_0(t)))}^{Q_{\zeta_1}(F_1(t,y_0(t)))} \dfrac{f_{\zeta_1}(v)}{f_1\left(t,Q_1\left(t,F_{\zeta_1}(v) \right) \right)} F_{Z^{(1)}}(t,v)\rd v
     \\[10pt]
     &\leq  \int_{\underset{i}{\min}Q_{\zeta_i}(F_i(t,y_0(t)))}^{\underset{i}{\max}Q_{\zeta_i}(F_i(t,y))} \left[\dfrac{f_{\zeta_2}(v)}{f_2\left(t,Q_2\left(t,F_{\zeta_2}(v) \right) \right)} F_{Z^{(2)}}(t,v) - \dfrac{f_{\zeta_1}(v)}{f_1\left(t,Q_1\left(t,F_{\zeta_1}(v) \right) \right)} F_{Z^{(1)}}(t,v)   \right]\rd v
    \end{split}
\end{equation}
for all $y\in[y_0(t),\max\{\overline{y_1},\overline{y_2}\}]$, $t\in(0,\infty)$.  By the same argument given in the proof of Proposition \ref{sosdprop1},
\begin{equation}
    \dfrac{f_{\zeta_2}(z)}{f_2\left(t,Q_2\left(t,F_{\zeta_2}(z) \right) \right)} F_{Z^{(2)}}(t,z) - \dfrac{f_{\zeta_1}(z)}{f_1\left(t,Q_1\left(t,F_{\zeta_1}(z) \right) \right)} F_{Z^{(1)}}(t,z)\leq F_{Z^{(2)}}(t,z)-F_{Z^{(1)}}(t,z)
\end{equation}
for all $z\in D_{\zeta}$, with strict inequality for at least one $z\in D_{\zeta}$, and all $t\in(0,\infty)$ if the inequalities (i)--(iii) hold for all $z\in D_{\zeta}$, with strict inequality for at least one $z\in D_{\zeta}$, and all $t\in(0,\infty)$.  Here, 
\begin{equation}\begin{split}
    0&\leq  \int_{z_0(t)}^{z} \left[\dfrac{f_{\zeta_2}(v)}{f_2\left(t,Q_2\left(t,F_{\zeta_2}(v) \right) \right)} F_{Z^{(2)}}(t,v) - \dfrac{f_{\zeta_1}(v)}{f_1\left(t,Q_1\left(t,F_{\zeta_1}(v) \right) \right)} F_{Z^{(1)}}(t,v)   \right]\rd v 
    \\
    &\leq  \int_{z_0(t)}^{z} \left[F_{Z^{(2)}}(t,v)-F_{Z^{(1)}}(t,v)\right] \rd v
\end{split}\end{equation}
where $z_0(t):=\underset{i}{\min}Q_{\zeta_i}(F_i(t,y_0(t)))$ for each $t\in(0,\infty)$ and $z:=\underset{i}{\max}Q_{\zeta_i}(F_i(t,y))$ for all $y\in D_Y$, and by the definition of SOSD, $Z_t^{(1)}\succsim_{SOSD}Z_t^{(2)}$ on $D_{\zeta}:=[z_0(t),\max\{\overline{z}_1,\overline{z}_2 \}]$ as required. In the case where $F_1(t,y)=F_{Y^{(1)}}(t,y)=F_{Y^{(2)}}(t,y)=F_2(t,y)$ for all $y\in[\min\{\underline{y_1},\underline{y_2}\},\max\{\overline{y_1},\overline{y_2}\}]$ and $t\in(0,\infty)$, we replace $F_{Z^{(i)}}(t,z)$ and $Q_i(t,u)$ with $F_{\zeta_i}(z)$ and $Q_{Y^{(i)}}(t,u)$ in the inequalities.
\end{proof}
\end{appendix}
\end{document}